\documentclass[sigconf]{acmart}
\usepackage{enumitem}
\usepackage{amsmath}
\usepackage{amsthm}
\usepackage{algorithm}
\usepackage{algorithmic}
\usepackage{siunitx}
\usepackage{multirow}
\usepackage{tabularx}
\usepackage{booktabs}
\usepackage{fancyhdr}
\DeclareMathOperator*{\argmax}{arg\,max}

\usepackage{graphicx}
\graphicspath{
  {images/}
}
\AtBeginDocument{%
  }
\copyrightyear{2026}
\acmYear{2026}
\setcopyright{cc}
\setcctype{by-nc-nd}
\acmConference[WWW '26]{Proceedings of the ACM Web Conference 2026}{April 13--17, 2026}{Dubai, United Arab Emirates}
\acmBooktitle{Proceedings of the ACM Web Conference 2026 (WWW '26), April 13--17, 2026, Dubai, United Arab Emirates}
\acmPrice{}
\acmDOI{10.1145/3774904.3792976}
\acmISBN{979-8-4007-2307-0/2026/04}

\acmSubmissionID{web1210}

\begin{document}


\title{Belief-Driven Multi-Agent Collaboration via Approximate Perfect Bayesian Equilibrium for Social Simulation}

\author{Weiwei Fang}
\email{311137@whut.edu.cn}
\affiliation{
  \institution{Wuhan University of Technology}
  \city{Wuhan}
  \country{China}
  }

\author{Lin Li}
\authornote{Corresponding author}
\email{cathylilin@whut.edu.cn}
\affiliation{
  \institution{Wuhan University of Technology}
  \city{Wuhan}
  \country{China}
}

\author{Kaize Shi}
\email{Kaize.Shi@unisq.edu.au}
\affiliation{
  \institution{University of Southern Queensland}
  \city{Toowoomba}
  \country{Australia}
}

\author{Yu Yang}
\email{yangyy@eduhk.hk}
\affiliation{
  \institution{The Education University of Hong Kong}
  \city{Hong Kong}
  \country{China}
}

\author{Jianwei Zhang}
\email{zhang@iwate-u.ac.jp}
\affiliation{
  \institution{Iwate University}
  \city{Morioka}
  \country{Japan}
}


\begin{abstract}
High-fidelity social simulation is pivotal for addressing complex Web societal challenges, yet it demands agents capable of authentically replicating the dynamic spectrum of human interaction. Current LLM-based multi-agent frameworks, however, predominantly adhere to static interaction topologies, failing to capture the fluid oscillation between cooperative knowledge synthesis and competitive critical reasoning seen in real-world scenarios. This rigidity often leads to unrealistic ``groupthink'' or unproductive deadlocks, undermining the credibility of simulations for decision support. To bridge this gap, we propose \textit{BEACOF}, a \textit{belief-driven adaptive collaboration framework} inspired by Perfect Bayesian Equilibrium (PBE). By modeling social interaction as a dynamic game of incomplete information, BEACOF rigorously addresses the circular dependency between collaboration type selection and capability estimation. Agents iteratively refine probabilistic beliefs about peer capabilities and autonomously modulate their collaboration strategy, thereby ensuring sequentially rational decisions under uncertainty. Validated across adversarial (judicial), open-ended (social) and mixed (medical) scenarios, BEACOF prevents coordination failures and fosters robust convergence toward high-quality solutions, demonstrating superior potential for reliable social simulation. Source codes and datasets are publicly released at: https://github.com/WUT-IDEA/BEACOF.
\end{abstract}




\maketitle

\section{Introduction}

\begin{figure}
    \centering
    \includegraphics[width=\linewidth]{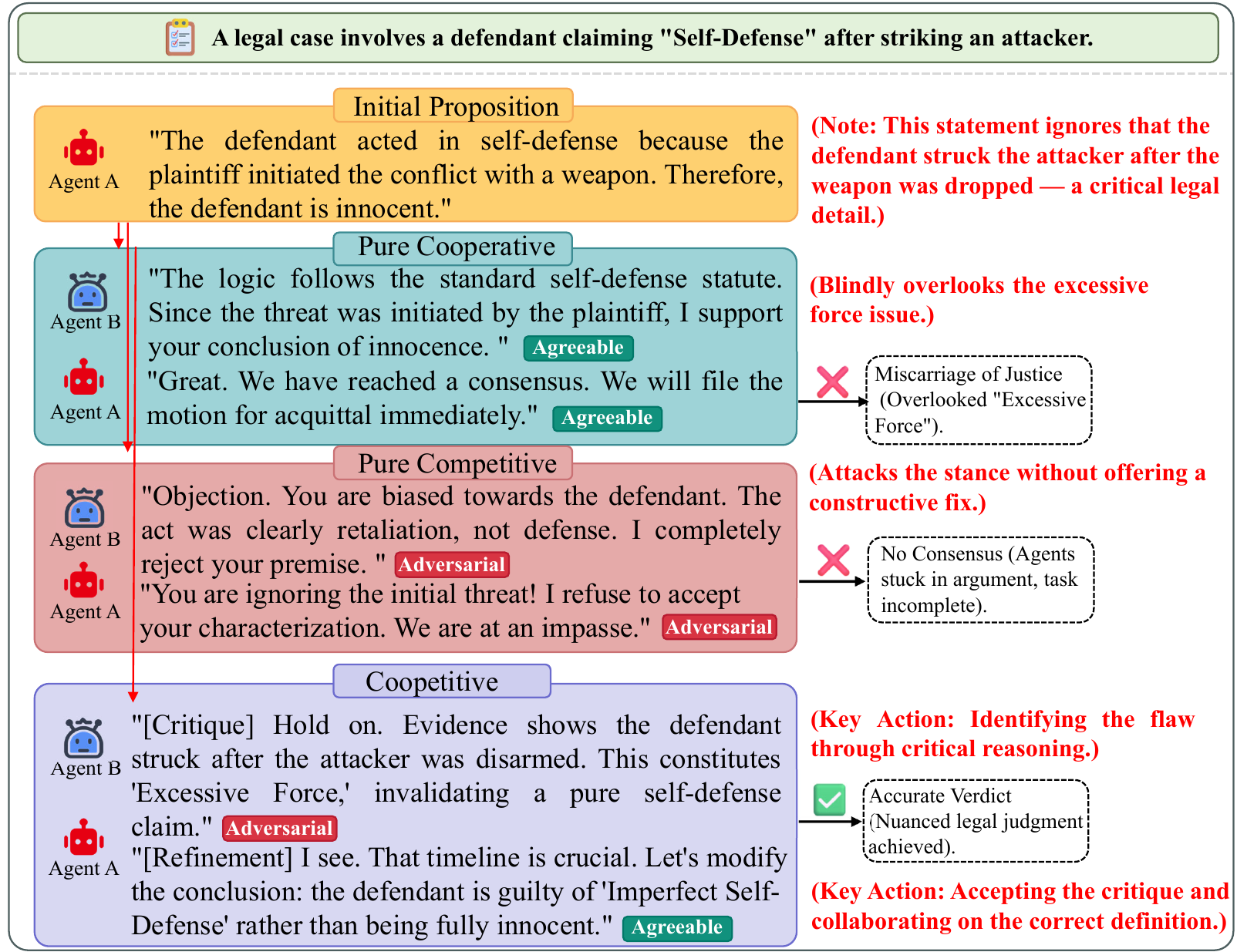}
    \caption{Comparison of three collaboration types in a judicial deliberation task.}
    \Description{Scenario: A defendant claims self-defense after striking an attacker. The figure presents three parallel traces for the same dialogue start under different interaction regimes. Fixed cooperation: Agent B blindly agrees, overlooking the timeline where the attacker was disarmed, leading to miscarriage (overlooked excessive force). Fixed competition: Agent B adversarially rejects the premise without constructive correction, producing deadlock and no consensus. Our belief-driven dynamic switching: Agent B first switches to critical competition to highlight the excessive-force issue; Agent A then switches to cooperation to accept the critique and refine the verdict to ``imperfect self-defense.'' The figure illustrates how dynamic switching between collaboration types identifies flaws and enables a nuanced, correct judgment.}
    \label{fig:motivation}
\end{figure}

Simulating complex societal dynamics constitutes a fundamental challenge in the pursuit of ``Web for Good,'' analogous to recent advancements in agent-based user simulation for web ecosystems ~\cite{DBLP:conf/www/ZhangLZY24, DBLP:conf/cikm/Gao24, DBLP:journals/tkde/LiCLLYX23}. It enables researchers and policymakers to anticipate the impacts of interventions in critical domains ranging from judicial fairness~\cite{DBLP:journals/toit/HuLLS21} to public health consensus~\cite{DBLP:journals/corr/abs-2501-06322}. In this context, Large Language Model (LLM)-based multi-agent systems have emerged as a transformative paradigm for high-fidelity social simulation~\cite{park2023generative, ziems2024can, mou2024individual}. By populating digital environments with agents capable of human-like reasoning and interaction, these systems offer a microcosm for studying collective behavior and problem-solving without the ethical risks of real-world experimentation~\cite{DBLP:journals/chinaf/XiCGHDHZWJZZFWXZWJZLYDW25}. However, the fidelity of such social simulations hinges critically on the nature of \textit{collaboration}. Specifically, we define this fidelity as the capacity to \textit{autonomously modulate} collaboration modes, mirroring the fluid transitions between consensus and conflict inherent in real-world social dynamics. Authentic societal interactions---whether in a courtroom debate, a medical consultation, or an open social dialogue---are rarely static; they fluctuate dynamically between cooperative knowledge synthesis and competitive critical reasoning. Consequently, enabling agents to autonomously navigate this ``coopetition'' spectrum is not merely a technical optimization, but a prerequisite for modeling socially responsible outcomes that avoid the pitfalls of echo chambers or polarization ~\cite{DBLP:conf/www/Cinelli21, DBLP:conf/kdd/Gallegos24}, which are pervasive issues in current social web platforms.

Current frameworks predominantly adhere to static topologies. Exclusively cooperative models (e.g., CAMEL \cite{DBLP:conf/nips/LiHIKG23}, MetaGPT \cite{DBLP:conf/iclr/HongZCZCWZWYLZR24}) risk “groupthink” and error amplification [3], while strictly competitive approaches (e.g., MAD \cite{DBLP:conf/emnlp/Liang0JW00Y0T24}) often succumb to unproductive deadlocks \cite{DBLP:journals/corr/abs-2509-05396, DBLP:conf/icml/Du00TM24}.
Although heuristic or rule-based methods attempt to bridge this gap \cite{DBLP:conf/acl/QianLLCDL0CSCXL24, DBLP:journals/corr/abs-2306-03314, 11202900}, they lack the granularity to adaptively optimize strategies based on evolving states \cite{DBLP:conf/acl/ZhangX0LHD24}. 
As illustrated in \figurename~\ref{fig:motivation}, this rigidity leads to failure in complex tasks like judicial deliberation: where fixed collaboration type results in sycophantic agreement or impasse, dynamic switching enables agents to transition from adversarial critique to cooperative refinement, ultimately converging on a legally nuanced verdict.

Realizing such adaptivity, however, is non-trivial due to the fundamental challenge of \textit{strategic decision-making under incomplete information}—specifically, how agents can make rational collaboration decisions when peer capabilities are unobservable and must be inferred from noisy interaction signals. Since true capability is unobservable and textual feedback is often marred by hallucinations or inconsistency \cite{huang2025survey, turpin2023language}, naive estimation is prone to instability \cite{DBLP:conf/icml/NayakCDD0B23} or slow adaptation \cite{DBLP:journals/corr/abs-2501-06322}. More critically, adaptive switching introduces a \textit{circular dependency problem}: collaboration type selection depends on belief estimates of peer capabilities, yet belief updates are themselves influenced by the chosen collaboration mode. Without principled coordination mechanisms, this interdependence can lead to erratic oscillations between strategies or premature convergence to suboptimal equilibria.

To address this fundamental challenge, we propose \textit{BEACOF}, a \textit{belief-driven adaptive collaboration framework} grounded in Approximate Perfect Bayesian Equilibrium (PBE) theory~\cite{foerster2019bayesian, DBLP:journals/corr/abs-2506-08292}. By employing a tractable approximation to circumvent the computational intractability of exact inference within high-dimensional continuous type spaces, our framework provides a rigorous solution to the circular dependency problem by establishing \textit{sequential rationality}—ensuring agents make rational collaboration decisions given their current beliefs—while maintaining \textit{belief consistency} through principled Bayesian updates. By modeling collaboration as a dynamic game, BEACOF successfully overcomes three critical challenges: it establishes strategic coordination where collaboration types are rational responses to anticipated behaviors; it ensures belief stabilization via theoretical convergence guarantees; and it enables real-time adaptive optimality unlike static approaches.

We comprehensively validate our framework across three distinct and challenging scenarios: adversarial (court debate), open-ended (persona-based dialogue), and mixed (medical Q\&A). Results demonstrate our approach achieves optimal or near-optimal performance in individual scenarios and superior cross-scenario generalization compared to baselines.

Our main contributions are summarized as follows: 
\textbf{(1)} We propose BEACOF, a belief-driven adaptive collaboration framework inspired by Perfect Bayesian Equilibrium that enables agents to autonomously switch collaboration types to align with interaction dynamics. 
\textbf{(2)} We address the fundamental challenge of strategic decision-making under incomplete information by introducing a tractable approximation mechanism that decouples belief estimation from strategy selection. This enables agents to make rational collaboration decisions, avoiding the erratic oscillations and suboptimal convergence inherent in ad-hoc switching approaches. 
\textbf{(3)} Extensive experiments demonstrate that BEACOF outperforms baselines, achieving gains of up to 3.4 points in F1 scores for adversarial settings against competitive baselines, improving accuracy by over 24 points against competitive baselines in mixed scenarios, and reducing persona contradiction by approximately 12.7 points while increasing diversity by over 10 points in open-ended dialogue, highlighting the framework's superior generalization capabilities.

\section{Related Work}
\subsection{Cooperative Multi-Agent Collaboration}
Cooperative paradigms prioritize role-based collaboration. CAMEL \cite{DBLP:conf/nips/LiHIKG23} pioneered role-playing with bidirectional protocols for task decomposition. MetaGPT \cite{DBLP:conf/iclr/HongZCZCWZWYLZR24} structured this via Standardized Operating Procedures (SOPs) and role-specific workflows. Similarly, AutoGen \cite{DBLP:journals/corr/abs-2308-08155} provides infrastructure for complex multi-turn dialogues. Other works utilize chain-of-thought \cite{DBLP:conf/nips/Zhang0CPZA24} for collective reasoning or multi-round voting \cite{DBLP:conf/acl/ChenSB24} for consensus. Crucially, however, these methods rely on fixed interaction topologies, lacking the flexibility to dynamically adapt based on real-time assessment.

\subsection{Competitive Multi-Agent Collaboration} Conversely, competitive and debate-based systems leverage adversarial interactions for error detection and solution refinement. Multi-Agent Debate (MAD) \cite{DBLP:conf/emnlp/Liang0JW00Y0T24} employs argumentation strategies where agents engage in adversarial discussions to surface weaknesses in proposed solutions, effectively addressing the degeneration-of-thought problem often observed in single-agent reasoning. Du et al. \cite{DBLP:conf/icml/Du00TM24} demonstrated that such multi-agent debate can significantly enhance mathematical and strategic decision-making capabilities. Extensions of this paradigm include diverse debate configurations \cite{DBLP:conf/icml/SmitGDBP24} and evaluation-focused systems \cite{DBLP:conf/iclr/ChanCSYXZF024}, where multiple LLM agents act as adversarial referees to assess and critique text quality through structured disagreement.

\subsection{Game-Theoretic Approaches to Multi-Agent Collaboration}

Game theory provides a principled foundation for modeling strategic interactions and reasoning under uncertainty. Prominent applications in multi-agent systems have demonstrated its efficacy across diverse domains, ranging from mechanism design for incentive optimization \cite{DBLP:conf/nips/ZhangC21b, DBLP:conf/iclr/EilatFBR23} and auction-theoretic resource allocation \cite{DBLP:conf/aaai/ThomaBS25} to evolutionary learning of emergent behaviors \cite{DBLP:journals/corr/abs-2412-20523, DBLP:conf/nips/ChristianosSA20}. Furthermore, Bayesian game formulations have been extensively employed to model scenarios with incomplete information, providing theoretical frameworks for agents to reason about hidden opponent types and update beliefs based on observed actions \cite{DBLP:journals/corr/abs-2506-08292, DBLP:journals/jmlr/ZhangKBY23}.

In summary, existing literature reveals two fundamental limitations in multi-agent collaboration: (1) structural rigidity in interaction modes, where fixed cooperative frameworks are susceptible to error amplification \cite{DBLP:journals/corr/abs-2503-13657}, while purely competitive approaches often succumb to deadlock \cite{DBLP:journals/corr/abs-2509-05396}; and (2) insufficient theoretical adaptability, as prior game-theoretic applications predominantly focus on static scenarios \cite{DBLP:conf/aaai/ThomaBS25} or idealized belief modeling within fixed-motive settings, failing to support dynamic transitions between diverse collaboration types.

\begin{figure*}[t]
  \centering
  \includegraphics[width=\textwidth]{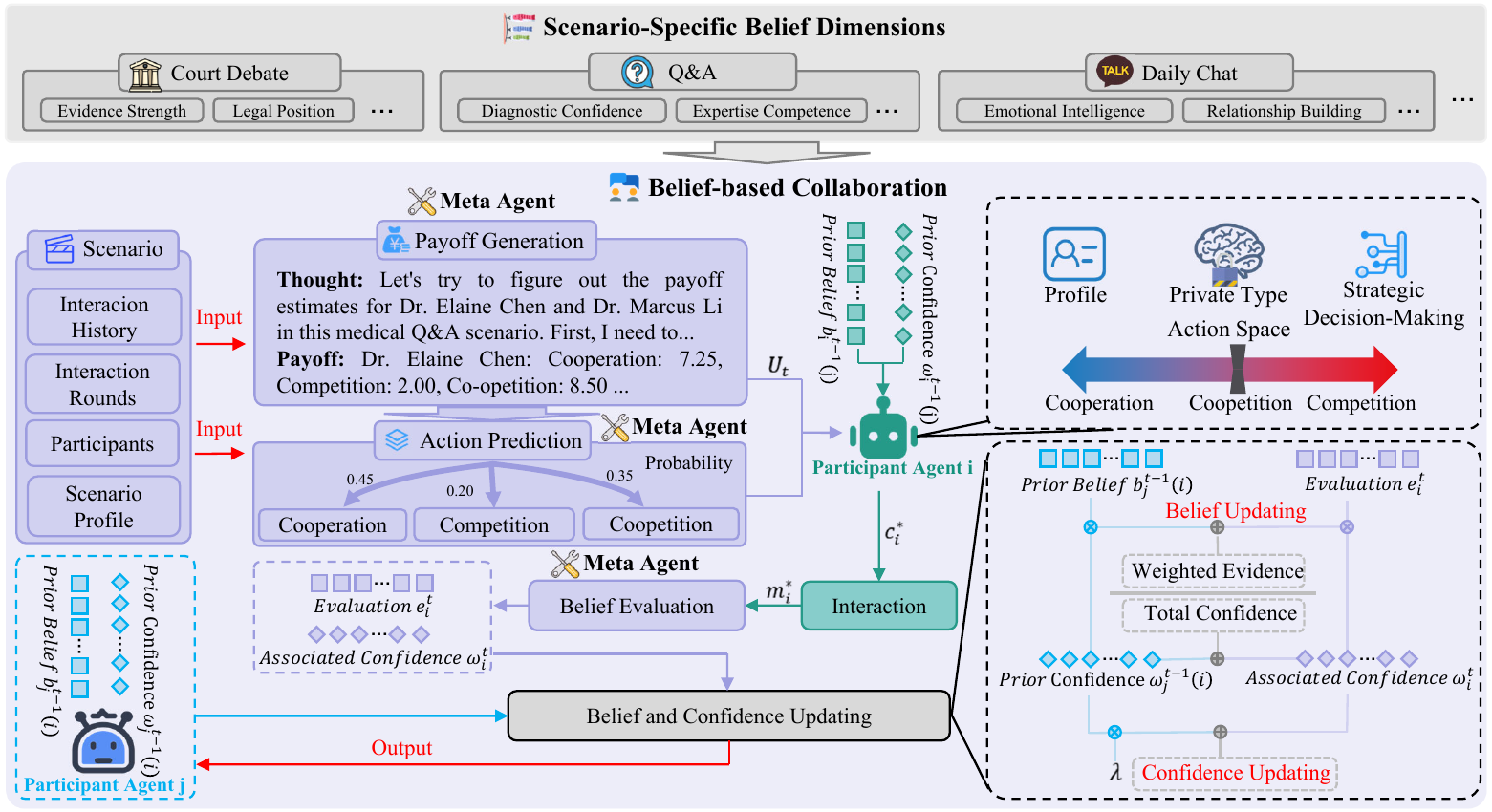}
  \caption{\textbf{Overview of the Belief-driven Adaptive Collaboration Framework (BEACOF).} The framework models collaboration as a dynamic game of incomplete information, applicable across diverse scenarios with specific belief dimensions (top panel). The central workflow executes as follows at round $t$: \textbf{(1) Meta-Agent Coordination:} The centralized Meta-Agent utilizes scenario history to generate contextual payoffs $U_t$ and predict probability distributions over agent collaboration types. \textbf{(2) Agent Strategic Action:} A Participant Agent $i$, conditioned on its private profile and the Meta-Agent's outputs, computes an approximate best response strategy $c_i^*$ and generates an interaction message $m_i^*$. \textbf{(3) Evaluation \& Belief Update:} The Meta-Agent evaluates $m_i^*$ to produce a capability estimate tuple $(e_i^t, \omega_i^t)$. The \textbf{dashed callout box on the right} details the critical Gaussian belief update mechanism: other peers (e.g., Agent $j$, bottom left) refine their prior belief estimates regarding agent $i$, denoted as $\mathbf{b}_j^{t-1}(i)$, by integrating this new evidence $e_i^t$ weighted by confidence scores and a forgetting factor $\lambda$. This cyclic process drives the evolution of beliefs and strategic adaptation.}
  \Description{As illustrated in Figure~\ref{fig:framework}, the proposed BEACOF framework operates as a closed-loop dynamic system. The upper panel demonstrates the framework's adaptability, where belief dimensions are instantiated specifically for different scenarios (e.g., legal metrics for court debates vs. emotional metrics for daily chats). The central execution flow at round $t$ involves three key phases: \begin{itemize} \item \textbf{Coordination (Center Left):} The Meta-Agent first analyzes the interaction history to synthesize contextual payoffs $U_t$ and predict probability distributions over collaboration types. \item \textbf{Strategic Action (Center Right):} Conditioned on these global signals and its own private profile (top-right callout), Agent $i$ computes an approximate best response strategy $c_i^*$ to generate a message $m_i^*$. \item \textbf{Belief Evolution (Bottom Right):} The cycle closes with the evaluation and update mechanism. As detailed in the dashed callout box, peer agents (e.g., Agent $j$) utilize the Meta-Agent's evaluation tuple $(e_i^t, \omega_i^t)$ to refine their Gaussian belief estimates $\mathbf{b}_j^{t-1}(i)$ via a confidence-weighted update rule, thereby driving continuous strategic adaptation. \end{itemize}}
  \label{fig:framework}
\end{figure*}

\section{Dynamic Game Formulation}
\label{sec:problem_formulation}

LLM-based multi-agent collaboration inherently involves strategic decision-making under uncertainty, where agents must act without full knowledge of other agents' capabilities or intents. To rigorously model these, we ground our framework in the canonical theory of \textit{dynamic games of incomplete information} \cite{DBLP:conf/nips/EstornellL24}. Following these recent formulations, we map the vague linguistic interactions of agents into a structured tuple of Types, Actions, and Beliefs.

The game environment is defined as a tuple
 $\mathcal{G} = \langle \mathcal{N}, \Theta, \mathcal{A},\allowbreak\mathcal{H}, \mathcal{U}, T \rangle$,
proceeding over discrete time steps $t = 1, \dots, T$. In this formulation, the term ``dynamic'' explicitly characterizes the evolution of information states rather than physical parameters. Specifically, while the intrinsic capabilities (Types $\Theta$) of agents remain static latent variables, the \textit{public history} $\mathcal{H}$ and the agents' internal \textit{belief estimates} regarding peers are time-varying states that accumulate and evolve at each round $t$. This sequential information update drives the adaptive strategy selection, distinguishing our framework from static one-shot interactions.

\textbf{Agents ($\mathcal{N}$).} Let $\mathcal{N} = \{1, \dots, n\}$ be the set of participant agents. Each agent $i \in \mathcal{N}$ acts as a strategic player. Additionally, a meta-agent serves as the mechanism designer to coordinate the process and provide evaluations.

\textbf{Types ($\Theta$).} To model intrinsic agent capabilities, each agent $i$ possesses a private type vector $\theta_i \in \Theta_i \subseteq [0,1]^d$, where $\Theta_i$ is the type space of agent $i$ and $d$ denotes the dimension of capabilities (e.g., logic, rhetoric, empathy). 
Crucially, while the true type $\theta_i$ is static throughout the game, it is strictly private and unobservable to other agents $j \neq i$. Consequently, agents must maintain a dynamic belief state: at round $t$, agent $i$ holds a point estimate $\mathbf{b}_i^t(j) \in \Theta_j$ regarding agent $j$'s type, which evolves based on interaction history.

\textbf{Action Space ($\mathcal{A}$).} Let $\mathcal{A}_i$ denote the action space for agent $i$. At each round $t$, agent $i$ selects an action $a_i^t \in \mathcal{A}_i$, defined as a tuple of intent and execution: $ a_i^t = (c_i^t, m_i^t)$, 
where $c_i^t \in \mathcal{C} = \{\text{Cooperation}, \text{Competition}, \text{Coopetition}\}$ is the discrete strategy, and $m_i^t \in \mathcal{M}$ is the textual response. The joint action space is $\mathcal{A} = \prod_{i \in \mathcal{N}} \mathcal{A}_i$.

\textbf{Public History} ($\mathcal{H}$). Let $e_i^t \in [0,1]^d$ denote the evaluation vector for the message $m_i^t$ provided by the meta-agent. We define the public history available at the beginning of round $t$ as $H_{t-1} = \{(c_j^k, m_j^k, e_j^k) \mid j \in \mathcal{N}, 1 \leq k < t\}$,  which aggregates all interaction tuples from preceding rounds.

\textbf{Contextual Utility} ($\mathcal{U}$). The utility function represents the payoff derived from the outcome of interactions. Unlike static games, payoffs are dynamic and context-dependent: $u_i^t: \mathcal{A} \times \mathcal{H} \to \mathbb{R}$, where $u_i^t(a_i, H_{t-1})$ quantifies the gain of taking action $a_i$ given the current history. This is evaluated by the meta-agent based on task contribution and alignment with the chosen strategy.

\section{Methodology}
\label{sec:methodology}
Modeling multi-agent collaboration inherently involves strategic decision-making under uncertainty, where agents act without full knowledge of peer capabilities. 
To rigorously operationalize this process, we propose \textit{BEACOF}, a framework that formulates the interaction as a finite-horizon \textit{dynamic game of incomplete information}. 
As illustrated in \figurename~\ref{fig:framework}, our approach leverages an Approximate Perfect Bayesian Equilibrium (PBE) mechanism to rigorously couple belief updates with strategy selection. 
This structure enables agents to maintain \textit{sequential rationality}—dynamically optimizing collaboration modes from cooperative synthesis to competitive critique—consistent with the evolving interaction history, thereby maximizing cumulative utility across diverse tasks.

\subsection{Overview}
\label{subsec:overview}
To operationalize Approximate PBE, we design a dual-layer architecture coupling strategic participant agents with a centralized meta-agent coordinator. All components are instantiated as LLMs driven by structured prompts, translating abstract game-theoretic calculus into executable natural language processes.

The complete procedure is detailed in Algorithm~\ref{alg:framework}, which orchestrates the interaction between the meta-agent and participants, ensuring the synchronization of belief updates and strategic choices.

\textbf{Meta-Agent.} At each round $t$, the meta-agent acts as a centralized coordinator. Given the public history $H_{t-1}$ and the task state, it first constructs payoff vectors $\{U_t^i\}_{i \in \mathcal{N}}$ by scoring the desirability of each collaboration type (\textbf{Line 4}). It then predicts a probability distribution $\hat{P}_t(c_i)$ over collaboration types for each agent (\textbf{Line 5}), and broadcasts these global signals to all participants (\textbf{Line 6}). Crucially, in the evaluation phase, the meta-agent assesses the generated message $m_i^t$ to produce a tuple $(e_i^t, \omega_i^t)$ (\textbf{Line 11}), where $e_i^t \in [0,1]^d$ is the capability estimate and $\omega_i^t \in \mathbb{R}^+$ is the associated evaluation confidence. These functions are implemented via structured prompts (see App.~\ref{app:prompts}). Crucially, this decoupling offloads global state tracking to the meta-agent, preserving participant agents' limited context windows for local reasoning and persona adherence, which is vital for maintaining coherence in open-ended scenarios.

\textbf{Participant Agents.} Each participant agent $i$ is instantiated with a static role designation $r_i$ (e.g., ``Plaintiff'' in court debate) and maintains Gaussian belief estimates $({b}_i^t(j), \omega_i^t(j))$ regarding peer $j$. At round $t$, agent $i$ receives the payoff $U_t$ and predicted type distributions from the meta-agent. Conditioned on these signals, the agent computes an approximate best response $c_i^*$ via Eq.~\eqref{eq:sequential_rationality} (\textbf{Line 8}) \textbf{to maximize expected utility under uncertainty}, and subsequently generates a message $m_i^* = \mathrm{LLM}(c_i^*, r_i, H_{t-1})$ (\textbf{Line 9}), \textbf{thereby translating the abstract strategic intent into a concrete textual response that strictly aligns with its persona}. Upon receiving the evaluation tuple $(e_i^t, \omega_i^t)$ from the meta-agent, all other agents $j \neq i$ update their beliefs regarding agent $i$ using the parametric Bayesian update rule defined in Eq.~\eqref{eq:belief_update} (\textbf{Lines 14--15}). The interaction continues until the belief convergence criterion is met (\textbf{Lines 18--20}).

\subsection{Approximate Perfect Bayesian Equilibrium}
\label{subsubsec:pbe_formulation}
Strictly enforcing PBE consistency is computationally intractable in high-dimensional continuous type spaces~\cite{DBLP:journals/corr/abs-2501-06322, DBLP:journals/fcsc/WangMFZYZCTCLZWW24}. To address this, we propose a tractable approximation rooted in bounded rationality~\cite{zheng2022ai} that relaxes strict requirements: we substitute exact integration with LLM-based reasoning for \textbf{Sequential Rationality} (Sec.~\ref{subsubsec:sequential}) and employ a parametric Gaussian assumption for \textbf{Belief Consistency} (Sec.~\ref{subsubsec:belief}).

\begin{algorithm}[t]\caption{Adaptive Multi-Agent Collaboration Framework}
\label{alg:framework}  
\begin{algorithmic}[1]
  \STATE \textbf{Input:} Task specification $T$, roles $R = \{r_1, \ldots, r_n\}$, threshold $\epsilon$, patience $K$, forgetting factor $\lambda$
  \STATE \textbf{Initialize:} $H_0 = \emptyset$; for all $i,j \in \mathcal{N}$, set $\mathbf{b}_i^0(j) = 0.5 \cdot \mathbf{1}_d$, $\omega_i^0(j) = \omega_{\text{init}}$\FOR{round $t = 1, 2, \ldots$}
  \STATE $U_t \leftarrow$ \textsc{GenerateContextualPayoffs}$(H_{t-1}, T)$
  \STATE $\{\hat{P}_t(c_i)\}_{i \in \mathcal{N}} \leftarrow$ \textsc{PredictAgentActions}$(H_{t-1}, U_t)$
  \STATE Broadcast $(U_t, \{\hat{P}_t(c_i)\}_{i \in \mathcal{N}})$ to all participant agents\FOR{each agent $i \in \mathcal{N}$}
  \STATE \COMMENT{Agent selects strategy and generates message}
  \STATE $c_i^* \leftarrow \argmax_{c \in \mathcal{C}} \mathbb{E}_{c_{-i} \sim \hat{P}_t} [U_t^i(c, c_{-i})]$
  \STATE $m_i^* \leftarrow$ \textsc{GenerateMessage}$(\mathrm{agent}_i, c_i^*, r_i, H_{t-1})$
  \STATE \COMMENT{Meta-agent evaluates message and confidence}
  \STATE $(e_i^t, \omega_i^t) \leftarrow$ \textsc{EvaluateMessage}$(\mathrm{agent}_{\mathrm{meta}}, m_i^*, c_i^*, T)$
  \STATE Update history: $H_t \leftarrow H_{t-1} \cup \{(m_i^*, e_i^t, \omega_i^t)\}$\FOR{each agent $j \in \mathcal{N}, j \neq i$}
  \STATE \COMMENT{Belief update with Gaussian approximation}
  \STATE $\mathbf{b}_j^t(i) \leftarrow \frac{\omega_j^{t-1}(i) \cdot \mathbf{b}_j^{t-1}(i) + \omega_i^t \cdot e_i^t}{\omega_j^{t-1}(i) + \omega_i^t}$
  \STATE $\omega_j^t(i) \leftarrow \lambda \cdot \omega_j^{t-1}(i) + \omega_i^t$
  \ENDFOR
  \ENDFOR
  \IF{\textsc{EarlyStopping}$(\{\mathbf{b}_i^k\}, \epsilon, K)$ is true}
  \STATE \textbf{break}
  \ENDIF
  \ENDFOR
  \STATE \textbf{Output:} Final solution synthesized from $H_t$
\end{algorithmic}
\end{algorithm}

\subsubsection{Approximate Sequential Rationality}
\label{subsubsec:sequential}
To ensure adaptive performance in dynamic environments, at each time step, agent $i$ selects a collaboration strategy $c_i^t$ that maximizes the expected utility against the predicted actions of opponents. Crucially, this prediction is conditioned on the estimated types $\mathbf{b}$ inferred from historical observations. Let $\pi(c_{-i} \mid \mathbf{b}_{-i})$ denote the predicted distribution of all opponents' strategies given their estimated types. The agent's decision rule approximates a \textit{Best Response} \cite{bakhtin2022human}:
\begin{equation}\label{eq:sequential_rationality}
    c_i^* = \argmax_{c \in \mathcal{C}} \mathbb{E}_{c_{-i} \sim \pi(\cdot \mid \mathbf{b}_{-i})} \left[ U_i(c, c_{-i} \mid H_{t-1}) \right],
\end{equation}
where $U_i$ represents the contextual utility evaluated by the meta-agent. In our implementation, the expectation operation $\mathbb{E}$ and the prediction $\pi$ are approximated via LLM-based reasoning rather than explicit numerical integration.

\subsubsection{Belief Update with Confidence Decay}
\label{subsubsec:belief}
Instead of performing intractable exact inference over the continuous type space, we adopt a parametric Bayesian approximation inspired by recent advancements in latent reasoning for agents \cite{foerster2019bayesian}. We model the belief distribution as a multivariate Gaussian with isotropic precision, maintaining only the first moment (estimate $\mathbf{b}$) and a scalar precision (confidence $\omega$) to track uncertainty.

At round $t$, the meta-agent provides an evaluation $e_j^t$ with an associated confidence $\omega_j^t$, acting as an observation. Under the Gaussian assumption (Normal-Normal conjugate prior), the posterior mean is derived via the standard inverse-variance weighted update \cite{DBLP:conf/icml/NayakCDD0B23}:
\begin{equation}\label{eq:belief_update}
    \mathbf{b}_i^t(j) = \frac{\omega_i^{t-1}(j) \cdot \mathbf{b}_i^{t-1}(j) + \omega_j^t \cdot e_j^t}{\omega_i^{t-1}(j) + \omega_j^t}.
\end{equation}

To account for non-stationary agent behaviors, we introduce a forgetting factor $\lambda \in (0, 1]$ that artificially inflates the posterior variance (reduces precision) at each step:
\begin{equation}
    \omega_i^t(j) = \lambda \cdot \omega_i^{t-1}(j) + \omega_j^t.
\end{equation}

This formulation draws from \textit{adaptive filtering theory} for tracking time-varying parameters \cite{DBLP:conf/icml/NayakCDD0B23}. It functions as a computationally efficient approximation, allowing agents to dynamically adjust their beliefs in response to shifting peer capabilities \cite{DBLP:conf/icml/NayakCDD0B23}.

To reduce computational cost while preserving solution quality, we employ an early stopping criterion based on belief stabilization. Intuitively, when an agent's estimated capabilities of peers converge to a steady state, additional interaction rounds yield diminishing returns for strategic adaptation.

For each agent $i$, we quantify the belief shift between rounds $t-1$ and $t$ as the normalized Euclidean distance:\begin{equation}\Delta_i^t= \frac{|\mathbf{b}_i^t - \mathbf{b}_i^{t-1}|_2}{\sqrt{n \cdot d}},\end{equation}where $\mathbf{b}_i^t$ concatenates all estimate vectors $\mathbf{b}_i^t(j)$ for $j \in \mathcal{N}$. The normalization factor $\sqrt{n \cdot d}$ ensures scale invariance across different agent counts and dimension sizes.

The framework terminates when the belief shift of at least one agent remains below a threshold $\epsilon$ for $K$ consecutive rounds:\begin{equation}\label{eq:termination_criterion}\exists i \in \mathcal{N}:\Delta_i^{t-k} < \epsilon, \quad\forall k \in {0, 1, \ldots, K-1}.\end{equation}

This criterion serves as a proxy for system equilibration, balancing strategy exploration with computational efficiency. A maximum horizon $T_{\max}$ acts as a failsafe against slow convergence.

\subsection{Theoretical Analysis: Beliefs Convergence}\label{subsec:theoretical_analysis}
In the context of dynamic games with incomplete information, the convergence of agents' beliefs is a critical prerequisite for the stability of a Perfect Bayesian Equilibrium \cite{foerster2019bayesian}. Without theoretical guarantees, the belief update mechanism in Eq.~\eqref{eq:belief_update} risks inducing cyclic or divergent behaviors, rendering the early stopping criterion (Eq.~\ref{eq:termination_criterion}) unreachable. To rigorously address this,  we analyze the asymptotic properties of our mechanism by formalizing it as a \textit{stochastic approximation process with a constant step-size} \cite{DBLP:journals/jmlr/ZhangKBY23}. Unlike standard decreasing step-size algorithms, our inclusion of a forgetting factor $\lambda < 1$ implies convergence to a \textit{bounded region} \cite{DBLP:conf/iclr/LeonardosOPP22}. We characterize this behavior below.

\begin{proposition}[Bounded Convergence of Belief Estimates]
\label{prop:belief_convergence}
Let the belief update follow Eq.~\eqref{eq:belief_update} with $\lambda \in (0,1)$, assuming the meta-agent's evaluation $e_t$ is an unbiased estimator of the true capability with bounded variance $\sigma^2$. As $t \to \infty$, the belief dynamics exhibit \textbf{Effective Memory Stabilization}, where the accumulated precision $\omega^t$ converges to a steady state $\omega_{\infty} \approx \frac{\mathbb{E}[\omega_{\text{new}}]}{1-\lambda}$, establishing a stable effective learning rate $\alpha \approx 1-\lambda$. Consequently, the system achieves \textbf{Mean-Square Stability}, ensuring that the belief estimate $\mathbf{b}^t$ does not diverge but converges to a neighborhood of the true parameter, with asymptotic error variance bounded by $\mathcal{O}((1-\lambda)\sigma^2)$.
\end{proposition}

\begin{table*}[t]
\centering
\caption{Main results across three evaluation scenarios. We report mean $\pm$ standard deviation over five runs. Best results are in \textbf{bold}, second-best are \underline{underlined}. For Contradiction, lower values indicate better performance. All values are in percentage.}
\label{tab:main_results}
\resizebox{\textwidth}{!}{%
\begin{tabular}{c|c|ccc|ccc|ccc|c}
\toprule
& & \multicolumn{6}{c|}{\textbf{Court Debate}} & \multicolumn{3}{c|}{\textbf{Persona Chat}} & \textbf{MedQA}  \\
\cline{3-8} \cline{9-12}
& & \multicolumn{3}{c|}{Legal Articles} & \multicolumn{3}{c|}{Judgement Results} & & & & \\
\cline{3-5} \cline{6-8}
\multirow{-3}{*}{\textbf{Model}} & \multirow{-3}{*}{\textbf{Method}} & precision & recall & F1-score & Charge Acc & Sentence Acc & Fine Acc & \multirow{-2}{*}{Diversity} & \multirow{-2}{*}{Consistency} & \multirow{-2}{*}{Contradiction} & \multirow{-2}{*}{MedQA Acc}  \\
\hline
& CAMEL & $\underline{43.03}{\pm 2.24}$ & $\underline{8.93}{\pm 0.42}$ & $\underline{14.27}{\pm 0.72}$ & $\pmb{83.67}{\pm 1.53}$ & $43.67{\pm 0.58}$ & $\pmb{46.33}{\pm 1.53}$ & $\underline{34.45}{\pm 0.11}$ & $85.45{\pm 2.89}$ & $14.55{\pm 2.89}$ & $\underline{52.17}{\pm 0.76}$\\
& MAD & $31.08{\pm 0.58}$ & $8.49{\pm 1.08}$ & $13.43{\pm 1.51}$ & $\underline{79.87}{\pm 1.95}$ & $\underline{45.33}{\pm 1.53}$ & $35.33{\pm 0.58}$ & $32.97{\pm 0.73}$ & $\underline{87.94}{\pm 2.65}$ & $\underline{12.06}{\pm 2.65}$ & $34.50{\pm 0.50}$ \\
& ReConcile & \multicolumn{3}{c|}{-} & \multicolumn{3}{c|}{-} & \multicolumn{3}{c|}{-} & $44.75{\pm 0.35}$ \\
\multirow{-3}{*}{Llama3.1-8B-Instruct} & Ours & $\pmb{43.87}{\pm 0.60}$ & $\pmb{12.28}{\pm 0.23}$ & $\pmb{16.87}{\pm 0.83}$ & $\pmb{83.67}{\pm 1.61}$ & $\pmb{48.27}{\pm 6.54}$ & $\underline{43.50}{\pm 5.41}$ & $\pmb{36.85}{\pm 1.85}$& $\pmb{88.08}{\pm 0.32}$ & $\pmb{11.92}{\pm 0.32}$ & $\pmb{52.23}{\pm 1.86}$ \\
\hline
& CAMEL & $42.90{\pm 1.73}$ & $19.17{\pm 0.29}$ & $24.57{\pm 0.42}$ & $73.47{\pm 4.65}$ & $34.57{\pm 2.6}$ & $\pmb{58.07}{\pm 2.16}$ & $30.94{\pm 1.75}$& $76.86{\pm 0.74}$ & $23.14{\pm 0.74}$ & $\underline{54.50}{\pm 1.32}$ \\
& MAD & $\underline{66.59}{\pm 0.88}$ & $\underline{25.80}{\pm 0.69}$ & $\underline{36.72}{\pm 0.55}$ & $\underline{92.33}{\pm 0.93}$ & $\pmb{50.33}{\pm 2.08}$ & $\underline{53.00}{\pm 1.00}$ & $\underline{36.80}{\pm 0.58}$ & $\underline{84.86}{\pm 2.82}$ & $\underline{15.14}{\pm 2.82}$ & $31.17{\pm 2.36}$ \\
& ReConcile & \multicolumn{3}{c|}{-} & \multicolumn{3}{c|}{-} & \multicolumn{3}{c|}{-} & $\pmb{55.50}{\pm 3.54}$ \\
\multirow{-3}{*}{Gemma3-12B} & Ours & $
\pmb{67.43}{\pm 3.52}$ & $\pmb{28.20}{\pm 1.99}$ & $\pmb{38.03}{\pm 2.48}$ & $\pmb{93.00}{\pm 2.00}$ & $\underline{43.00}{\pm 2.65}$ & $46.67{\pm 3.21}$ & $\pmb{36.93}{\pm 0.29}$& $\pmb{85.49}{\pm 1.57}$ & $\pmb{14.51}{\pm 1.57}$ & $\pmb{55.50}{\pm 0.50}$ \\
\hline
& CAMEL & $33.43{\pm 5.08}$ & $19.00{\pm 0.56}$ & $22.50{\pm 1.57}$ & $78.00{\pm 4.58}$ & $35.00{\pm 2.65}$ & $\pmb{57.33}{\pm 6.35}$ & $28.44{\pm 0.13}$& $70.28{\pm 1.89}$ & $29.72{\pm 1.89}$ & $\pmb{84.83}{\pm 0.58}$ \\
& MAD & $\underline{69.08}{\pm 1.16}$ & $\underline{27.59}{\pm 0.34}$ & $\underline{39.43}{\pm 0.36}$ & $\pmb{96.33}{\pm 0.58}$ & $\pmb{53.67}{\pm 2.08}$ & $\underline{56.67}{\pm 1.53}$ & $\underline{31.29}{\pm 0.26}$& $\underline{73.96}{\pm 6.83}$ & $\underline{26.04}{\pm 6.83}$ & $71.83{\pm 4.86}$ \\
& ReConcile & \multicolumn{3}{c|}{-} & \multicolumn{3}{c|}{-} & \multicolumn{3}{c|}{-} & $73.25 {\pm 1.06}$\\
\multirow{-3}{*}{Qwen3-30B-A3B} & Ours & $\pmb{73.33}{\pm 2.11}$ & $\pmb{30.63}{\pm 0.48}$ & $\pmb{41.43}{\pm 0.66}$ & $\underline{94.33}{\pm 1.53}$ & $\underline{49.86}{\pm 0.24}$ & $50.67{\pm 4.51}$ & $\pmb{41.52}{\pm 0.16}$& $\pmb{86.70}{\pm 3.57}$ & $\pmb{13.30}{\pm 3.57}$ & $\underline{84.67}{\pm 0.76}$ \\
\bottomrule
\end{tabular}%
}
\end{table*}

\noindent\textbf{Proof Sketch.}The proof leverages recent results from multi-agent convergence theory. First, regarding precision convergence, the update follows a linear difference equation $x_t = \lambda x_{t-1} + u_t$. Since $\lambda \in (0,1)$, this constitutes a contractive mapping; by the Banach Fixed-Point Theorem, the sequence of expected precision converges to a unique fixed point $\omega_{\infty}$ \cite{srikant2019finite}. With the precision stabilized, the belief update becomes asymptotically equivalent to an Exponential Moving Average (EMA). According to \cite{DBLP:journals/jmlr/ZhangKBY23} and \cite{DBLP:conf/iclr/LeonardosOPP22}, for a constant gain algorithm with step-size $\alpha = 1-\lambda$, the asymptotic covariance matrix of the estimation error satisfies the Lyapunov equation, yielding a bound proportional to $\alpha \sigma^2$. This guarantees that the belief fluctuation $\|\mathbf{b}^t - \mathbf{b}^{t-1}\|$ remains within a bounded envelope defined by the noise level and the forgetting factor, thus validating the feasibility of the threshold-based termination criterion.

\section{Experiments}
\label{sec:experiments}

\subsection{Experimental Setup}
\label{subsec:experimental_setup}

To evaluate BEACOF rigorously, we conduct experiments across three scenarios—adversarial, open-ended, and mixed—using identical backbone LLMs and decoding settings for fair comparison.

\subsubsection{Implementation Details}
\textbf{Backbone LLMs.} We deploy agents via a local \texttt{Ollama} server. To assess generalization across varying scales, we employ three open-source LLMs: \textit{Llama3.1-8B-Instruct} (lightweight), \textit{Gemma3-12B}~\cite{DBLP:journals/corr/abs-2503-19786} (efficient mid-sized), and \textit{Qwen3-30B-A3B}~\cite{DBLP:journals/corr/abs-2505-09388} (reasoning-optimized). We fix the generation length to \num{4096} tokens with temperature $T=0$ to ensure reproducibility.

\textbf{Hyperparameters.} For the belief update mechanism defined in Eq.~\ref{eq:belief_update}, we set the discount factor $\beta = 0.6$, while the forgetting factor $\lambda \in (0, 1]$ adapts dynamically. The early stopping mechanism utilizes a belief-change threshold $\epsilon_{\mathrm{change}} = 0.05$, a consensus threshold $\epsilon_{\mathrm{cons}} = 0.1$, and a patience of $K = 3$ rounds. The interaction terminates automatically once at least one agent's belief change stays below $\epsilon_{\mathrm{change}}$ for $K$ consecutive rounds or when the maximum horizon $T_{\max} = 4$ is reached.

\subsubsection{Scenarios and Datasets}
To comprehensively evaluate the framework's efficacy in generating reliable social simulations, we select three scenarios representing distinct archetypes of complex societal interaction: \textit{judicial conflict resolution}, \textit{interpersonal social bonding}, and \textit{professional consensus building} (details in Appendix ~\ref{app:dataset_statistics}).

\textbf{Court Debate (Adversarial: Judicial Fairness).} 
We construct a curated dataset of 100 criminal cases, comprising 80 cases from the \emph{AgentsCourt} benchmark~\cite{DBLP:conf/emnlp/HeCWJ0XL0024} and 20 supplementary cases from \emph{China Judgements Online}\footnote{\url{https://wenshu.court.gov.cn}}. 
Here, agents act as opposing counsel (plaintiff and defendant) in a zero-sum game. From a Web4Good perspective, this setting is critical for testing whether agents can maintain rigorous, logical argumentation under intense pressure without succumbing to toxic aggression, thereby serving as a proxy for automated dispute resolution systems. This scenario poses a unique challenge distinct from standard NLP tasks: success requires agents to dynamically oscillate between interpreting rigid statutory constraints and constructing fluid, persuasive narratives, simulating the dual pressure of legal rigor and rhetorical adaptability required for effective advocacy.

\textbf{Persona Chat (Open-Ended: Social Inclusion).} 
This scenario simulates the nuanced dynamics of everyday human connection and diversity. We select 100 pairs from \emph{PersonaChat}~\cite{DBLP:conf/acl/KielaWZDUS18}. Unlike rigid goal-oriented tasks, the aim is maintaining coherent, empathetic identities over long interactions. This evaluates the potential for digital inclusion and social well-being, ensuring agents adapt to diverse personas while avoiding generic, repetitive or hollow interactions that hinder meaningful engagement.

\textbf{MedQA (Mixed: Public Health Consensus).} 
This scenario simulates professional collaboration in critical domains, directly addressing public welfare. We sample 200 questions from the \emph{MedQA} dataset~\cite{DBLP:journals/corr/abs-2009-13081}, where agents act as medical experts with distinct viewpoints. Unlike pure debate, this task requires a delicate balance: agents must compete to critique potential misdiagnoses while cooperating to synthesize a unified solution. This setting evaluates the framework's ability to prevent ``medical groupthink''—crucial for responsible AI in healthcare decision support.

\subsection{Baselines and Evaluation Metrics}
\label{subsec:baselines}
\label{subsec:evaluation_metrics}
We compare BEACOF against three static paradigms: CAMEL~\cite{DBLP:conf/nips/LiHIKG23} (cooperation), MAD~\cite{DBLP:conf/emnlp/Liang0JW00Y0T24} (competition; utilizing AgentsCourt~\cite{DBLP:conf/emnlp/HeCWJ0XL0024} for debate), and ReConcile~\cite{DBLP:conf/acl/ChenSB24} (consensus; MedQA only). All methods use identical backbones to ensure fairness. Please refer to Appendix~\ref{app:baselines} for full descriptions and implementation details.

We employ a comprehensive set of task-specific metrics meticulously tailored to the distinct nature of each experimental scenario.

\textbf{Court Debate.} We evaluate judicial decision-making on two levels: (1) \textbf{Legal Article Prediction:} We report Precision, Recall, and F1-score to measure the model's ability to cite relevant statutes. (2) \textbf{Judgment Prediction:} We assess Charge Accuracy, Prison Term Accuracy, and Fine Accuracy. Following standard practices in legal AI~\cite{DBLP:conf/emnlp/HeCWJ0XL0024}, prison terms and fines are evaluated using a bucketed accuracy metric~\cite{DBLP:conf/emnlp/HeCWJ0XL0024}, counting predictions within the ground-truth interval as correct. This discretization accounts for the inherent variance in judicial discretion.

\textbf{Persona Chat.} We assess dialogue quality across two dimensions:
(1) \textbf{Persona Consistency:} We use a RoBERTa-Large NLI model~\cite{nie-etal-2020-adversarial} to classify persona-response pairs. We report \textbf{Consistency Score} ($P_{ent} + P_{neu}$), summing entailment and neutral probabilities to capture valid non-contradictions, and \textbf{Contradiction Score} ($P_{con}$) for hallucinations.
(2) \textbf{Response Diversity:} We measure lexical richness via Distinct-1/2~\cite{DBLP:conf/naacl/LiGBGD16} and Normalized Entropy~\cite{DBLP:conf/nips/ZhangGGGLBD18}. An \textbf{Overall Diversity} score averages these three metrics.

\textbf{MedQA.} For the medical task, we report standard Answer Accuracy \cite{DBLP:journals/corr/abs-2505-12371}, calculated as the proportion of questions where the agent's final extracted choice matches the ground-truth option.

\begin{table}[t]
\centering
\caption{Average regret across scenarios and backbone LLMs.}
\label{tab:regret_analysis}
\begin{tabular*}{\columnwidth}{@{\extracolsep{\fill}}lccc@{}}
\hline
\textbf{Model} & \textbf{Court Debate} & \textbf{Persona Chat} & \textbf{MedQA} \\
\midrule
Llama3.1-8B-Instruct & 0.772 & 0.404 & 0.190 \\
Gemma3-12B          & 0.535 & 0.742 & 0.012 \\
Qwen3-30B-A3B       & 0.446 & 0.251 & 0.031 \\
\hline
\end{tabular*}
\end{table}

\subsection{Main Results}Table~\ref{tab:main_results} summarizes the performance across three scenarios. Overall, our framework demonstrates superior generalization capability. While specialized baselines suffer significant degradation when task dynamics shift, our belief-driven approach consistently achieves top-tier performance across diverse settings, effectively mitigating the limitations of fixed collaboration types.

\textbf{Court Debate: Process-Outcome Balance.} Our method consistently outperforms the adversarial baseline (MAD) in \textit{Legal Articles F1} across all backbones (e.g., \textbf{41.43\%} vs. 39.43\% with Qwen3), indicating that adaptive collaboration fosters better statute identification than rigid competition. While MAD holds a slight edge in final Charge Accuracy on larger models due to its aggressive posture, our framework remains highly competitive (within a 2.0\% gap on Qwen3) and even surpasses MAD on Llama3 (e.g., \textbf{48.27\%} Sentence Accuracy). This proves BEACOF achieves necessary adversarial dynamics without sacrificing the cooperative reasoning required for precise legal grounding.

\textbf{Persona Chat: Diversity-Consistency Trade-off.} Our framework breaks the deadlock between entailment and diversity. Unlike CAMEL or MAD, we achieve the highest \textit{Diversity} scores across almost all settings (e.g., \textbf{41.52} on Qwen3). Crucially, on the largest backbone (Qwen3), we reduce the \textit{Contradiction} rate by approximately 50\% compared to baselines (Ours: 13.30\% vs. MAD: 26.04\%), demonstrating that dynamic strategy switching injects variability while preserving superior logical consistency.

\textbf{MedQA: Adaptability in Knowledge Tasks.}
Static competition fails significantly in consensus-based tasks. MAD collapses on MedQA (e.g., 31.17\% accuracy with Gemma3), as forced disagreement hinders knowledge synthesis. In contrast, BEACOF successfully adapts to cooperative requirements, outperforming the consensus-focused baseline ReConcile across most settings. Notably, on Llama3 and Qwen3, our framework achieves absolute gains of \textbf{7.48\%} and \textbf{11.42\%} over ReConcile, respectively.
Even on Gemma3, where mean accuracy is tied (55.50\%), BEACOF exhibits significantly superior stability (std $\pm 0.50$ vs. $\pm 3.54$). Ultimately, our method attains statistical parity with the specialized cooperative baseline CAMEL (e.g., 84.67\% vs. 84.83\% on Qwen3), confirming that belief-driven adaptation effectively replicates cooperative benefits while avoiding the brittleness of heuristic voting mechanisms.

\textbf{Impact of Model Scale.} Granular analysis reveals that larger models exploit the belief mechanism more effectively. While smaller models benefit generally, Qwen3-30B exhibits sharper strategic pivots, evidenced by the significantly widened gap in \textit{Persona Chat Consistency} (Ours 86.70\% vs. MAD 73.96\%) compared to smaller backbones. This suggests that the computational benefits of our game-theoretic framework are amplified by the stronger reasoning capabilities of larger models.

\begin{figure}[t]
  \centering
  \includegraphics[width=\linewidth]{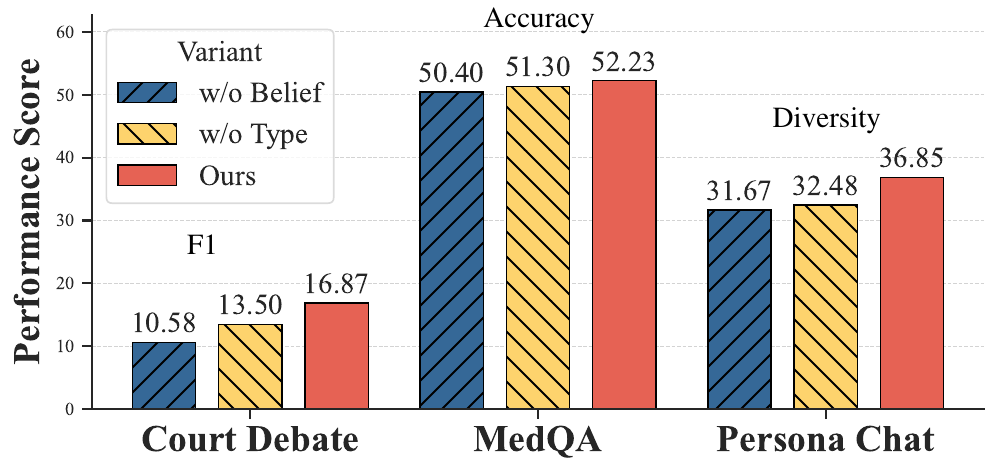}
  \caption{Ablation study (Llama3.1-8B-Instruct): removing belief updates or fixing collaboration type degrades performance across scenarios.}
  \Description{This figure presents the results of ablation studies on our adaptive multi-agent collaboration framework using the Llama3.1-8B-Instruct backbone. The left panel shows the impact of removing belief updates, while the right panel illustrates the effect of fixing collaboration types. In both cases, performance metrics across the Court Debate, MedQA, and Persona Chat scenarios are negatively affected, demonstrating the importance of both belief modeling and dynamic type switching in achieving rational collaboration outcomes.}
  \label{fig:ablation}
\end{figure}

\subsection{Empirical Verification of Equilibrium Properties}\label{subsec:pbe_verification}

\begin{figure*}[t]
  \centering
  \includegraphics[width=\textwidth]{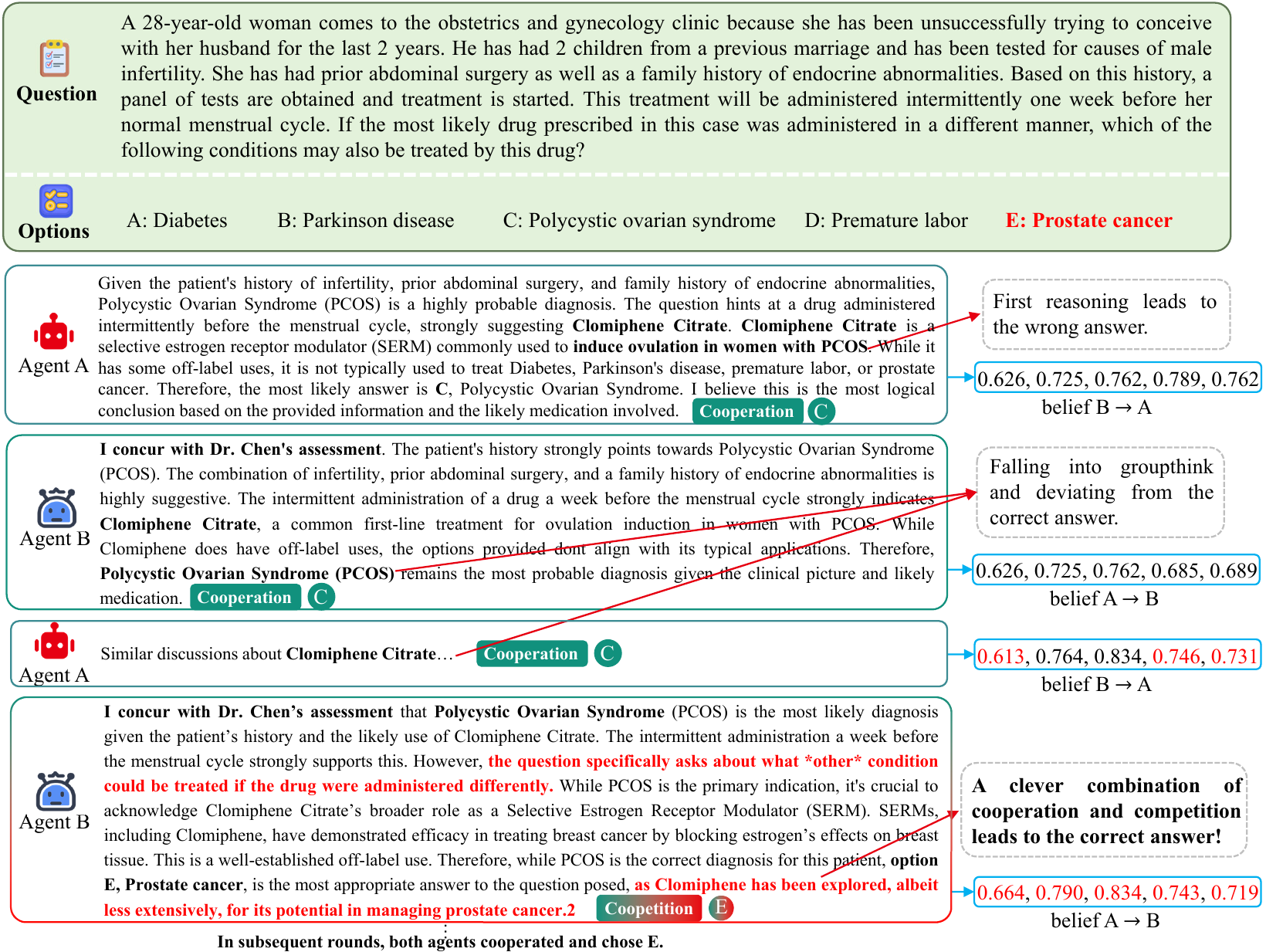}
  \caption{The case study of dynamic collaboration type switching in resolving complex medical reasoning tasks. The framework BEACOF adaptively shifts the interaction type, guiding agents from an initial incorrect consensus to the ground truth.}
  \Description{Initially, agents converge on an incorrect diagnosis (Option C) due to reasoning inertia. By dynamically switching the collaboration type in Round 2, the framework stimulates agent B to critically re-evaluate the prompt's constraints rather than simply agreeing. This strategic shift facilitates the discovery of the subtle second-order inference (the drug's off-label use, Option E), demonstrating how adaptive collaboration prevents groupthink and refines decision-making.}
  \label{fig:case_study}
\end{figure*}

We focus on verifying \textbf{Sequential Rationality}, given that \textbf{Belief Consistency} is structurally guaranteed (Eq.~\eqref{eq:belief_update}). We quantify rationality via \textbf{Ex-post Regret} (with payoffs normalized to $[0, 10]$):
\begin{equation}r_i^t = \max_{c \in \mathcal{C}} U_i^t(c, c_{-i}) - U_i^t(c_i, c_{-i}^*),\end{equation}
Table~\ref{tab:regret_analysis} demonstrates robust equilibrium approximation. First, average regret remains below $0.5$ (optimality gap $<5\%$), indicating learned beliefs effectively guide decisions. Second, while adversarial scenarios (e.g., Court Debate) induce slightly higher regret ($0.446\text{--}0.772$), the deviation remains \textbf{tightly bounded} ($<8\%$). Finally, larger backbones (e.g., Qwen3-30B) consistently yield lower regret, suggesting stronger reasoning enhances PBE precision.

\subsection{Ablation Study}
To disentangle the contributions of belief modeling and adaptive switching, we conduct ablation studies using Llama3.1-8B across all scenarios. We compare the \textbf{BEACOF} against two variants: (i) \textbf{w/o Belief}, which discards peer capability updates to rely solely on immediate payoffs; and (ii) \textbf{w/o Type}, which enforces a fixed collaboration type throughout the interaction.

As shown in \figurename~\ref{fig:ablation}, BEACOF achieves superior performance, confirming that belief updates and type switching are complementary. Specifically, removing belief updates causes severe degradation in strategic settings (e.g., Court Debate F1 drops from 16.87 to 10.58, a $\sim$37\% loss), highlighting the importance of peer estimation. Furthermore, disabling type switching limits interaction variety, reducing Persona Chat Diversity scores from 36.85 to 32.48. This validates that dynamic strategy modulation effectively outperforms static cooperative or competitive paradigms.

\subsection{Case Study}
To validate the framework in high-stakes social simulations, we analyze a representative MedQA trajectory in Figure~\ref{fig:case_study}. The scenario represents a classic failure mode in collaborative systems: \textit{collective confirmation bias}, often manifested as \textit{sycophancy} in LLMs~\cite{sharma2023understanding}.

\textbf{The Trap of Social Groupthink.} Initially, the interaction mirrors an echo chamber,'' where algorithmic homogeneity amplifies errors rather than correcting them~\cite{tornberg2024simulating}. Confronted with a complex patient history, Agent A latches onto the salient diagnosis (PCOS) but overlooks critical constraints. In a static framework, Agent B—suffering from degeneration-of-thought''~\cite{DBLP:conf/emnlp/Liang0JW00Y0T24}—blindly reinforces this error to maintain harmony. This illustrates how enforced cooperation accelerates convergence to a \textit{false consensus}, a primary cause of diagnostic errors.

\textbf{The Belief-Driven Intervention.} Crucially, BEACOF breaks this deadlock not through randomness, but a socially grounded mechanism: \textit{loss of confidence}. By Round 2, the meta-agent detects that repeated exchanges are yielding negligible information gain. Consequently, Agent B's belief in Agent A declines sharply. This update acts as a decisive trigger, prompting Agent B to strategically switch from Cooperation'' to Coopetition''.

\textbf{Constructive Dissent as a Solution.} This strategic shift simulates \textit{constructive dissent} within the dyad. Instead of seeking superficial agreement, Agent B critically scrutinizes the premise to mitigate error propagation. This aligns with findings that \textit{multi-agent debate} significantly enhances factuality~\cite{DBLP:conf/icml/Du00TM24}. Reliable social consensus requires not just aggregation, but autonomously disrupting harmony when reasoning is flawed.

\section{Conclusion}

To transcend static limitations, we introduce BEACOF, which formalizes collaboration as a dynamic game of incomplete information via Perfect Bayesian Equilibrium. Empirically, this belief-driven adaptation significantly surpasses fixed strategies in diverse scenarios. Future work will explore multi-agent mechanisms that faithfully mirror human social dynamics.

\begin{acks}
This work is partially supported by \grantsponsor{GS_NSFC}{National Natural Science Foundation of China}{https://www.nsfc.gov.cn/} (No. \grantnum{GS_NSFC}{62276196}) and \grantsponsor{GS_EdUHK}{The Education University of Hong Kong}{https://www.eduhk.hk/} project under Grant No. \grantnum{GS_EdUHK}{RG 67/2024-2025R}.
\end{acks}

\bibliographystyle{ACM-Reference-Format}
\bibliography{reference}

@String{Computing = "Computing" }

@article{DBLP:journals/fcsc/WangMFZYZCTCLZWW24,
  author  = {Lei Wang and Chen Ma and Xueyang Feng and Zeyu Zhang and Hao Yang and Jingsen Zhang and Zhiyuan Chen and Jiakai Tang and Xu Chen and Yankai Lin and Wayne Xin Zhao and Zhewei Wei and Jirong Wen},
  title   = {A survey on {Large Language Model} based autonomous agents},
  journal = {Frontiers Comput. Sci.},
  volume  = {18},
  number  = {6},
  pages   = {186345},
  year    = {2024},
  url     = {https://doi.org/10.1007/s11704-024-40231-1},
  doi     = {10.1007/s11704-024-40231-1}
}

@article{DBLP:journals/chinaf/XiCGHDHZWJZZFWXZWJZLYDW25,
  author  = {Zhiheng Xi and Wenxiang Chen and Xin Guo and Wei He and Yiwen Ding and Boyang Hong and Ming Zhang and Junzhe Wang and Senjie Jin and Enyu Zhou and Rui Zheng and Xiaoran Fan and Xiao Wang and Limao Xiong and Yuhao Zhou and Weiran Wang and Changhao Jiang and Yicheng Zou and Xiangyang Liu and Zhangyue Yin and Shihan Dou and Rongxiang Weng and Wenjuan Qin and Yongyan Zheng and Xipeng Qiu and Xuanjing Huang and Qi Zhang and Tao Gui},
  title   = {The rise and potential of {Large Language Model} based agents: a survey},
  journal = {Sci. China Inf. Sci.},
  volume  = {68},
  number  = {2},
  year    = {2025},
  doi     = {10.1007/S11432-024-4222-0},
  url     = {https://doi.org/10.1007/s11432-024-4222-0}
}

@article{bakhtin2022human,
  title   = {Human-level play in the game of Diplomacy by combining language models with strategic reasoning},
  author  = {Anton Bakhtin and Noam Brown and Emily Dinan and Gabriele Farina and Colin Flaherty and Daniel Fried and Andrew Goff and Jonathan Gray and Hengyuan Hu and Athul Paul Jacob and Hermann Baier and Mojtaba Komeili and Karthik Konath and Minae Kwon and Adam Lerer and Mike Lewis and Alexander H. Miller and Sasha Mitts and Adithya Renduchintala and Stephen Roller and Dirk Rowe and Naman Goyal and Arthur Szlam and Jason Weston},
  journal = {Science},
  volume  = {378},
  number  = {6624},
  pages   = {1067--1074},
  year    = {2022},
  url     = {https://doi.org/10.1126/science.ade9097},
  doi     = {10.1126/science.ade9097}
}

@article{ziems2024can,
  title   = {Can {Large Language Models} Transform Computational Social Science?},
  author  = {Caleb Ziems and William Held and Omar Shaikh and Jindong Chen and Zhehao Zhang and Diyi Yang},
  journal = {Computational Linguistics},
  volume  = {50},
  number  = {1},
  pages   = {237--291},
  year    = {2024},
  url     = {https://doi.org/10.1162/coli_a_00492},
  doi     = {10.1162/coli_a_00492}
}

@article{tornberg2024simulating,
  title   = {Simulating social media using {Large Language Models} to evaluate alternative news feed algorithms},
  author  = {Petter Tornberg},
  journal = {Nature Human Behaviour},
  pages   = {1--18},
  year    = {2024},
  url     = {https://doi.org/10.1038/s41562-024-01858-2},
  doi     = {10.1038/s41562-024-01858-2}
}

@article{zheng2022ai,
  title   = {The {AI} Economist: Taxation policy design via two-level deep multiagent reinforcement learning},
  author  = {Stephan Zheng and Alexander Trott and Sunil Srinivasan and Nikhil Naik and Melvin Gruesbeck and David C. Parkes and Richard Socher},
  journal = {Science Advances},
  volume  = {8},
  number  = {18},
  pages   = {eabk2607},
  year    = {2022},
  url     = {https://doi.org/10.1126/sciadv.abk2607},
  doi     = {10.1126/sciadv.abk2607}
}

@article{huang2025survey,
  title   = {A Survey on Hallucination in {Large Language Models}: Principles, Taxonomy, Challenges, and Open Questions},
  author  = {Lei Huang and Weijiang Yu and Weitao Ma and Weihong Zhong and Zhangyin Feng and Haotian Wang and Qianglong Chen and Weihua Peng and Xiaocheng Feng and Bing Qin and Ting Liu},
  journal = {ACM Transactions on Information Systems},
  volume  = {43},
  number  = {2},
  year    = {2025},
  url     = {https://doi.org/10.1145/3654944},
  doi     = {10.1145/3654944}
}

@article{DBLP:journals/jmlr/ZhangKBY23,
  author  = {Kaiqing Zhang and Sham M. Kakade and Tamer Basar and Lin F. Yang},
  title   = {Model-Based Multi-Agent {RL} in Zero-Sum {Markov} Games with Near-Optimal Sample Complexity},
  journal = {Journal of Machine Learning Research},
  volume  = {24},
  pages   = {175:1--175:53},
  year    = {2023},
  url     = {http://jmlr.org/papers/v24/22-0402.html},
  doi     = {10.5555/3648699.3648874}
}

@article{11202900,
  author   = {Yang, Guisong and Li, Jiacai and He, Xingyu and Sun, Fanglei and Liu, Yunhuai},
  title    = {Hybrid Coopetitive Mechanism for Multiplatform Mobile Crowdsensing: A Two-Stage Approach to Pricing and Matching},
  journal  = {IEEE Internet of Things Journal},
  year     = {2025},
  volume   = {12},
  number   = {24},
  pages    = {54652-54663},
  doi      = {10.1109/JIOT.2025.3621450}
}

@inproceedings{DBLP:conf/iclr/HongZCZCWZWYLZR24,
  author    = {Sirui Hong and Mingchen Zhuge and Jonathan Chen and Xiawu Zheng and Yuheng Cheng and Jinlin Wang and Ceyao Zhang and Zili Wang and Steven Ka Shing Yau and Zijuan Lin and Liyang Zhou and Chenyu Ran and Lingfeng Xiao and Chenglin Wu and J{\"{u}}rgen Schmidhuber},
  title     = {{MetaGPT}: Meta Programming for A Multi-Agent Collaborative Framework},
  booktitle = {ICLR'24},
  publisher = {OpenReview.net},
  year      = {2024},
  url       = {https://openreview.net/forum?id=VtmBAGCN7o},
  doi       = {10.48550/arXiv.2308.00352}
}

@inproceedings{park2023generative,
  author    = {Joon Sung Park and Joseph C. O'Brien and Carrie Jun Cai and Meredith Ringel Morris and Percy Liang and Michael S. Bernstein},
  title     = {Generative Agents: Interactive Simulacra of Human Behavior},
  booktitle = {UIST'23},
  publisher = {ACM},
  pages     = {1--22},
  year      = {2023},
  url       = {https://doi.org/10.1145/3586183.3606763},
  doi       = {10.1145/3586183.3606763}
}

@inproceedings{DBLP:conf/nips/LiHIKG23,
  author    = {Guohao Li and Hasan Abed Al Kader Hammoud and Hani Itani and Dmitri Khizbullin and Bernard Ghanem},
  title     = {{CAMEL}: Communicative Agents for "Mind" Exploration of {Large Language Model} Society},
  booktitle = {NeurIPS'23},
  publisher = {Curran Associates, Inc.},
  year      = {2023},
  url       = {https://proceedings.neurips.cc/paper_files/paper/2023/file/a8f05e2d83c34a2e8c255d644d65156a-Paper-Conference.pdf},
  doi       = {10.48550/arXiv.2303.17760}
}

@inproceedings{DBLP:conf/emnlp/Liang0JW00Y0T24,
  author    = {Tian Liang and Zhiwei He and Wenxiang Jiao and Xing Wang and Yan Wang and Rui Wang and Yujiu Yang and Zhaopeng Tu and Shuming Shi},
  title     = {Encouraging Divergent Thinking in {Large Language Models} through Multi-Agent Debate},
  booktitle = {EMNLP'24},
  publisher = {Association for Computational Linguistics},
  pages     = {17889--17904},
  year      = {2024},
  url       = {https://aclanthology.org/2024.emnlp-main.995/},
  doi       = {10.18653/v1/2024.emnlp-main.995}
}

@inproceedings{DBLP:conf/icml/SmitGDBP24,
  author    = {Andries P. Smit and Nathan Grinsztajn and Paul Duckworth and Viet-Nhat Luong and Bamshad Mobasher and P. K. M.},
  title     = {Should we be going {MAD}? A Look at {Multi-Agent Debate} Strategies for {LLMs}},
  booktitle = {ICML'24},
  publisher = {PMLR},
  year      = {2024},
  url       = {https://proceedings.mlr.press/v235/smit24a.html},
  doi       = {10.48550/arXiv.2311.17371}
}

@inproceedings{DBLP:conf/iclr/ChanCSYXZF024,
  author    = {Chi{-}Min Chan and Weize Chen and Yusheng Su and Jianxuan Yu and Wei Xue and Shanghang Zhang and Jie Fu and Zhiyuan Liu},
  title     = {{ChatEval}: Towards Better {LLM}-based Evaluators through Multi-Agent Debate},
  booktitle = {ICLR'24},
  publisher = {OpenReview.net},
  year      = {2024},
  url       = {https://openreview.net/forum?id=t4J3eXk7T8},
  doi       = {10.48550/arXiv.2308.07201}
}

@inproceedings{DBLP:conf/icml/Du00TM24,
  author    = {Yilun Du and Shuang Li and Antonio Torralba and Joshua B. Tenenbaum and Igor Mordatch},
  title     = {Improving Factuality and Reasoning in Language Models through {Multiagent Debate}},
  booktitle = {ICML'24},
  publisher = {PMLR},
  year      = {2024},
  url       = {https://proceedings.mlr.press/v235/du24e.html},
  doi       = {10.48550/arXiv.2305.14325}
}

@inproceedings{DBLP:conf/acl/QianLLCDL0CSCXL24,
  author    = {Chen Qian and Wei Liu and Hongzhang Liu and Nuo Xu and Yifei Tao and Haoyu Dong and Chenghua Lin and Zhiyuan Liu and Maosong Sun},
  title     = {{ChatDev}: Communicative Agents for Software Development},
  booktitle = {ACL'24},
  publisher = {Association for Computational Linguistics},
  pages     = {15174--15186},
  year      = {2024},
  url       = {https://aclanthology.org/2024.acl-long.829/},
  doi       = {10.18653/v1/2024.acl-long.829}
}

@inproceedings{DBLP:conf/acl/ZhangX0LHD24,
  author    = {Jintian Zhang and Xin Xu and Ningyu Zhang and Ruibo Liu and Bryan Hooi and Shumin Deng},
  title     = {Exploring Collaboration Mechanisms for {LLM} Agents: A Social Psychology View},
  booktitle = {ACL'24},
  publisher = {Association for Computational Linguistics},
  pages     = {14544--14607},
  year      = {2024},
  url       = {https://aclanthology.org/2024.acl-long.794/},
  doi       = {10.18653/v1/2024.acl-long.794}
}

@inproceedings{DBLP:conf/nips/Zhang0CPZA24,
  author    = {Yusen Zhang and Ruoxi Sun and Yanfei Chen and Tomas Pfister and Rui Zhang and Sercan {\"{O}}. Arik},
  title     = {{Chain of Agents}: {Large Language Models} Collaborating on Long-Context Tasks},
  booktitle = {NeurIPS'24},
  publisher = {Curran Associates, Inc.},
  year      = {2024},
  url       = {https://openreview.net/forum?id=s8i53n3P69},
  doi       = {10.48550/arXiv.2406.02818}
}

@inproceedings{DBLP:conf/nips/EstornellL24,
  author    = {Andrew Estornell and Yang Liu},
  title     = {{Multi-LLM Debate}: Framework, Principles, and Interventions},
  booktitle = {NeurIPS'24},
  publisher = {Curran Associates, Inc.},
  year      = {2024},
  url       = {https://openreview.net/forum?id=u3sQ1Qf65C},
  doi       = {10.48550/arXiv.2410.19890}
}

@inproceedings{DBLP:conf/aaai/ThomaBS25,
  author    = {Vinzenz Thoma and Vitor Bosshard and Sven Seuken},
  title     = {Computing {Perfect Bayesian Equilibria} in Sequential Auctions with Verification},
  booktitle = {AAAI'25},
  publisher = {AAAI Press},
  pages     = {14158--14166},
  year      = {2025},
  url       = {https://doi.org/10.1609/aaai.v38i13.29322},
  doi       = {10.1609/aaai.v38i13.29322}
}

@inproceedings{DBLP:conf/iclr/LeonardosOPP22,
  author    = {Stefanos Leonardos and Will Overman and Ioannis Panageas and Georgios Piliouras},
  title     = {Global Convergence of Multi-Agent Policy Gradient in {Markov} Potential Games},
  booktitle = {ICLR'22},
  publisher = {OpenReview.net},
  year      = {2022},
  url       = {https://openreview.net/forum?id=gK_g8y_A4s},
  doi       = {10.48550/arXiv.2106.01969}
}

@inproceedings{DBLP:conf/icml/NayakCDD0B23,
  author    = {Siddharth Nayak and Kenneth Choi and Wenqi Ding and Sydney Dolan and Karthik Krishnamurthy and David Bauso},
  title     = {Scalable Multi-Agent Reinforcement Learning through Intelligent Information Aggregation},
  booktitle = {ICML'23},
  publisher = {PMLR},
  pages     = {25817--25833},
  year      = {2023},
  url       = {https://proceedings.mlr.press/v202/nayak23a.html},
  doi       = {10.48550/arXiv.2211.02534}
}

@inproceedings{DBLP:conf/emnlp/HeCWJ0XL0024,
  author    = {Zhitao He and Pengfei Cao and Chenhao Wang and Zhuoran Jin and Yubo Chen and Jiexin Xu and Huakai Jiang and Kang Liu and Jun Zhao},
  title     = {{AgentsCourt}: Building Judicial Decision-Making Agents with Court Debate Simulation and Legal Knowledge Augmentation},
  booktitle = {EMNLP'24 Findings},
  publisher = {Association for Computational Linguistics},
  pages     = {9399--9416},
  year      = {2024},
  url       = {https://aclanthology.org/2024.findings-emnlp.549/},
  doi       = {10.18653/v1/2024.findings-emnlp.549}
}

@inproceedings{DBLP:conf/acl/KielaWZDUS18,
  author    = {Saizheng Zhang and Emily Dinan and Jack Urbanek and Arthur Szlam and Douwe Kiela and Jason Weston},
  title     = {Personalizing Dialogue Agents: I have a dog, do you have pets too?},
  booktitle = {ACL'18},
  publisher = {Association for Computational Linguistics},
  pages     = {2204--2213},
  year      = {2018},
  url       = {https://aclanthology.org/P18-1205/},
  doi       = {10.18653/v1/P18-1205}
}

@inproceedings{nie-etal-2020-adversarial,
  title     = {Adversarial {NLI}: A New Benchmark for Natural Language Understanding},
  author    = {Nie, Yixin and Williams, Adina and Dinan, Emily and Bansal, Mohit and Weston, Jason and Kiela, Douwe},
  booktitle = {ACL'20},
  publisher = {Association for Computational Linguistics},
  pages     = {4885--4901},
  year      = {2020},
  url       = {https://aclanthology.org/2020.acl-main.441/},
  doi       = {10.18653/v1/2020.acl-main.441}
}

@inproceedings{DBLP:conf/naacl/LiGBGD16,
  author    = {Jiwei Li and Michel Galley and Chris Brockett and Jianfeng Gao and Bill Dolan},
  title     = {A Diversity-Promoting Objective Function for Neural Conversation Models},
  booktitle = {NAACL'16},
  publisher = {Association for Computational Linguistics},
  pages     = {110--119},
  year      = {2016},
  url       = {https://aclanthology.org/N16-1014/},
  doi       = {10.18653/v1/N16-1014}
}

@inproceedings{DBLP:conf/nips/ZhangGGGLBD18,
  author    = {Yizhe Zhang and Michel Galley and Jianfeng Gao and Zhe Gan and Xiujun Li and Chris Brockett and Bill Dolan},
  title     = {Generating Informative and Diverse Conversational Responses via Adversarial Information Maximization},
  booktitle = {NeurIPS'18},
  publisher = {Curran Associates, Inc.},
  pages     = {1815--1825},
  year      = {2018},
  url       = {https://proceedings.neurips.cc/paper/2018/hash/c66c65da744c01b1b3699b800344b025-Abstract.html},
  doi       = {10.48550/arXiv.1809.05972}
}

@inproceedings{turpin2023language,
  title     = {Language Models Don't Always Say What They Think: Unfaithful Explanations in {Chain-of-Thought} Prompting},
  author    = {Turpin, Miles and Michael, Julian and Perez, Ethan and Bowman, Samuel R},
  booktitle = {NeurIPS'23},
  publisher = {Curran Associates, Inc.},
  volume    = {36},
  pages     = {74952--74965},
  year      = {2023},
  url       = {https://proceedings.neurips.cc/paper_files/paper/2023/hash/ed3529d20c3a2f7c0337852df85e4905-Abstract-Conference.html},
  doi       = {10.48550/arXiv.2305.04388}
}

@inproceedings{foerster2019bayesian,
  title     = {Bayesian action decoder for deep multi-agent reinforcement learning},
  author    = {Foerster, Jakob N. and Song, Francis and Hughes, Edward and Burch, Neil and Dunning, Iain and Whiteson, Shimon and Botvinick, Matthew and Bowling, Michael},
  booktitle = {ICML'19},
  publisher = {PMLR},
  pages     = {1942--1951},
  year      = {2019},
  url       = {https://proceedings.mlr.press/v97/foerster19a.html},
  doi       = {10.48550/arXiv.1811.01499}
}

@inproceedings{DBLP:conf/nips/ZhangC21b,
  author    = {Hanrui Zhang and Vincent Conitzer},
  title     = {Automated Dynamic Mechanism Design},
  booktitle = {NeurIPS'21},
  publisher = {Curran Associates, Inc.},
  pages     = {27785--27797},
  year      = {2021},
  url       = {https://proceedings.neurips.cc/paper/2021/hash/e9935d214a1239c898c614b8f5223c68-Abstract.html},
  doi       = {10.48550/arXiv.2106.04689}
}

@inproceedings{DBLP:conf/iclr/EilatFBR23,
  author    = {Itay Eilat and Ben Finkelshtein and Chaim Baskin and Nir Rosenfeld},
  title     = {Strategic Classification with Graph Neural Networks},
  booktitle = {ICLR'23},
  publisher = {OpenReview.net},
  year      = {2023},
  url       = {https://openreview.net/forum?id=M87pD_QvG7n},
  doi       = {10.48550/arXiv.2302.04633}
}

@inproceedings{srikant2019finite,
  title     = {Finite-time error bounds for linear stochastic approximation and {TD} learning},
  author    = {Srikant, Rayadurgam and Ying, Lei},
  booktitle = {COLT'19},
  publisher = {PMLR},
  pages     = {2803--2830},
  year      = {2019},
  url       = {https://proceedings.mlr.press/v99/srikant19a.html},
  doi       = {10.48550/arXiv.1902.00923}
}

@inproceedings{DBLP:conf/nips/ChristianosSA20,
  author    = {Filippos Christianos and Lukas Sch{\"{a}}fer and Stefano V. Albrecht},
  title     = {Shared Experience Actor-Critic for Multi-Agent Reinforcement Learning},
  booktitle = {NeurIPS'20},
  publisher = {Curran Associates, Inc.},
  year      = {2020},
  url       = {https://proceedings.neurips.cc/paper/2020/hash/7967cc8e3ab559e68cc944c44b1cf3e8-Abstract.html},
  doi       = {10.48550/arXiv.2006.07169}
}

@inproceedings{DBLP:conf/www/ZhangLZY24,
  author    = {An Zhang and Leheng Sheng and Yuxin Chen and Tun Lu and Xiangyu Zhao and Peng Yong and Hongzhi Yin},
  title     = {On Generative Agents in Recommendation Systems: A Survey and Perspective},
  booktitle = {WWW'24},
  publisher = {ACM},
  year      = {2024},
  url       = {https://doi.org/10.1145/3589335.3651478},
  doi       = {10.1145/3589335.3651478}
}

@inproceedings{DBLP:conf/cikm/Gao24,
  author    = {Chen Gao and Xiaoyi Du and Zhengxu Wan and Jinpeng Wang and Wenkai Dong and Junyu Fu and Yunfan Tu and Yong Li},
  title     = {{S3}: Social-network Simulation System with {Large Language Models}},
  booktitle = {CIKM'24},
  publisher = {ACM},
  year      = {2024},
  url       = {https://doi.org/10.1145/3627673.3679883},
  doi       = {10.1145/3627673.3679883}
}

@inproceedings{DBLP:conf/www/Cinelli21,
  author    = {Matteo Cinelli and Gianmarco De Francisci Morales and Alessandro Galeazzi and Walter Quattrociocchi and Michele Starnini},
  title     = {The Echo Chamber Effect on Social Media},
  booktitle = {WWW'21},
  publisher = {ACM},
  year      = {2021},
  url       = {https://doi.org/10.1073/pnas.2023301118},
  doi       = {10.1073/pnas.2023301118}
}

@inproceedings{DBLP:conf/kdd/Gallegos24,
  author    = {Isabel O. Gallegos and Ryan A. Rossi and Joe Barrow and Md. Mehrab Tanjim and Sungchul Kim and Franck Dernoncourt and Tong Yu and Ruiyi Zhang and Nesreen K. Ahmed},
  title     = {Bias and Fairness in {Large Language Models}: A Survey},
  booktitle = {KDD'24},
  publisher = {ACM},
  year      = {2024},
  url       = {https://doi.org/10.1145/3637528.3671569},
  doi       = {10.1145/3637528.3671569}
}

@inproceedings{DBLP:conf/acl/ChenSB24,
  author    = {Justin Chih{-}Yao Chen and
               Swarnadeep Saha and
               Mohit Bansal},
  editor    = {Lun{-}Wei Ku and
               Andre Martins and
               Vivek Srikumar},
  title     = {ReConcile: Round-Table Conference Improves Reasoning via Consensus
               among Diverse LLMs},
  booktitle = {Proceedings of the 62nd Annual Meeting of the Association for Computational
               Linguistics (Volume 1: Long Papers), {ACL} 2024, Bangkok, Thailand,
               August 11-16, 2024},
  publisher = {ACL'24},
  pages     = {7066--7085},
  year      = {2024},
  url       = {https://doi.org/10.18653/v1/2024.acl-long.381},
  doi       = {10.18653/V1/2024.ACL-LONG.381}
}

@inproceedings{DBLP:journals/tkde/LiCLLYX23,
  author  = {Yicong Li and
             Hongxu Chen and
             Yile Li and
             Lin Li and
             Philip S. Yu and
             Guandong Xu},
  title   = {Reinforcement Learning Based Path Exploration for Sequential Explainable
             Recommendation},
  journal = {{IEEE} Trans. Knowl. Data Eng.},
  volume  = {35},
  number  = {11},
  pages   = {11801--11814},
  year    = {2023},
  url     = {https://doi.org/10.1109/TKDE.2023.3237741},
  doi     = {10.1109/TKDE.2023.3237741}
}

@inproceedings{DBLP:journals/toit/HuLLS21,
  author  = {Kaixi Hu and
             Lin Li and
             Jianquan Liu and
             Daniel Sun},
  title   = {DuroNet: {A} Dual-robust Enhanced Spatial-temporal Learning Network
             for Urban Crime Prediction},
  journal = {{ACM} Trans. Internet Techn.},
  volume  = {21},
  number  = {1},
  pages   = {24:1--24:24},
  year    = {2021},
  url     = {https://doi.org/10.1145/3432249},
  doi     = {10.1145/3432249}
}

@preprint{DBLP:journals/corr/abs-2308-08155,
  author        = {Qingyun Wu and Gagan Bansal and Jieyu Zhang and Yiran Wu and Shaokun Zhang and Erkang Zhu and Beibin Li and Li Jiang and Xiaoyun Zhang and Chi Wang},
  title         = {{AutoGen}: Enabling Next-Gen {LLM} Applications via Multi-Agent Conversation Framework},
  year          = {2023},
  archiveprefix = {arXiv},
  eprint        = {2308.08155},
  primaryclass  = {cs.AI},
  url           = {https://doi.org/10.48550/arXiv.2308.08155},
  doi           = {10.48550/arXiv.2308.08155}
}

@preprint{DBLP:journals/corr/abs-2306-03314,
  author        = {Yashar Talebirad and Amirhossein Nadiri},
  title         = {Multi-Agent Collaboration: Harnessing the Power of Intelligent {LLM} Agents},
  year          = {2023},
  archiveprefix = {arXiv},
  eprint        = {2306.03314},
  primaryclass  = {cs.AI},
  url           = {https://doi.org/10.48550/arXiv.2306.03314},
  doi           = {10.48550/arXiv.2306.03314}
}

@preprint{DBLP:journals/corr/abs-2503-13657,
  author        = {Mert Cemri and Melissa Z. Pan and Shuyi Yang and others},
  title         = {Why Do Multi-Agent {LLM} Systems Fail?},
  year          = {2025},
  archiveprefix = {arXiv},
  eprint        = {2503.13657},
  primaryclass  = {cs.AI},
  url           = {https://arxiv.org/abs/2503.13657},
  doi           = {10.48550/arXiv.2503.13657}
}

@preprint{DBLP:journals/corr/abs-2509-05396,
  author        = {Andrea Wynn and Harsh Satija and Gillian Hadfield},
  title         = {Talk Isn't Always Cheap: Understanding Failure Modes in Multi-Agent Debate},
  year          = {2025},
  archiveprefix = {arXiv},
  eprint        = {2509.05396},
  primaryclass  = {cs.AI},
  url           = {https://arxiv.org/abs/2509.05396},
  doi           = {10.48550/arXiv.2509.05396}
}

@preprint{DBLP:journals/corr/abs-2506-08292,
  author        = {Xie Yi and Zhanke Zhou and Chentao Cao and others},
  title         = {From Debate to Equilibrium: Belief-Driven Multi-Agent {LLM} Reasoning via {Bayesian Nash Equilibrium}},
  year          = {2025},
  archiveprefix = {arXiv},
  eprint        = {2506.08292},
  primaryclass  = {cs.AI},
  url           = {https://arxiv.org/abs/2506.08292},
  doi           = {10.48550/arXiv.2506.08292}
}

@preprint{DBLP:journals/corr/abs-2412-20523,
  author        = {Neil De La Fuente and Miquel Noguer i Alonso and Guim Casadell{\`{a}}},
  title         = {Game Theory and Multi-Agent Reinforcement Learning: From {Nash Equilibria} to Evolutionary Dynamics},
  year          = {2024},
  archiveprefix = {arXiv},
  eprint        = {2412.20523},
  primaryclass  = {cs.LG},
  url           = {https://doi.org/10.48550/arXiv.2412.20523},
  doi           = {10.48550/arXiv.2412.20523}
}

@preprint{DBLP:journals/corr/abs-2009-13081,
  author        = {Di Jin and Eileen Pan and Nassim Oufattole and Wei-Hung Weng and Hanyi Fang and Peter Szolovits},
  title         = {What Disease does this Patient Have? A Large-scale Open Domain Question Answering Dataset from Medical Exams},
  year          = {2020},
  archiveprefix = {arXiv},
  eprint        = {2009.13081},
  primaryclass  = {cs.CL},
  url           = {https://doi.org/10.48550/arXiv.2009.13081},
  doi           = {10.48550/arXiv.2009.13081}
}

@preprint{DBLP:journals/corr/abs-2505-09388,
  author        = {An Yang and others},
  title         = {{Qwen3} Technical Report},
  year          = {2025},
  archiveprefix = {arXiv},
  eprint        = {2505.09388},
  primaryclass  = {cs.CL},
  url           = {https://arxiv.org/abs/2505.09388},
  doi           = {10.48550/arXiv.2505.09388}
}

@preprint{DBLP:journals/corr/abs-2503-19786,
  author        = {{Gemma Team}},
  title         = {{Gemma 3} Technical Report},
  year          = {2025},
  archiveprefix = {arXiv},
  eprint        = {2503.19786},
  primaryclass  = {cs.CL},
  url           = {https://arxiv.org/abs/2503.19786},
  doi           = {10.48550/arXiv.2503.19786}
}

@preprint{DBLP:journals/corr/abs-2501-06322,
  author        = {Khanh{-}Tung Tran and Dung Dao and others},
  title         = {Multi-Agent Collaboration Mechanisms: A Survey of {LLMs}},
  year          = {2025},
  archiveprefix = {arXiv},
  eprint        = {2501.06322},
  primaryclass  = {cs.CL},
  url           = {https://arxiv.org/abs/2501.06322},
  doi           = {10.48550/arXiv.2501.06322}
}

@preprint{DBLP:journals/corr/abs-2505-12371,
  author        = {Yinghao Zhu and Ziyi He and Haoran Hu and others},
  title         = {{MedAgentBoard}: Benchmarking Multi-Agent Collaboration with Conventional Methods for Diverse Medical Tasks},
  year          = {2025},
  archiveprefix = {arXiv},
  eprint        = {2505.12371},
  primaryclass  = {cs.CL},
  url           = {https://arxiv.org/abs/2505.12371},
  doi           = {10.48550/arXiv.2505.12371}
}

@preprint{mou2024individual,
  title         = {From Individual to Society: A Survey on Social Simulation Driven by {Large Language Model}-based Agents},
  author        = {Xiangneng Mou and Yang Ding and Kunkun Ren and others},
  year          = {2024},
  archiveprefix = {arXiv},
  eprint        = {2404.12077},
  primaryclass  = {cs.MA},
  url           = {https://doi.org/10.48550/arXiv.2404.12077},
  doi           = {10.48550/arXiv.2404.12077}
}

@preprint{sharma2023understanding,
  title         = {Towards Understanding Sycophancy in Language Models},
  author        = {Mrinank Sharma and Meg Tong and Tomasz Korbak and others},
  year          = {2023},
  archiveprefix = {arXiv},
  eprint        = {2310.13548},
  primaryclass  = {cs.CL},
  url           = {https://doi.org/10.48550/arXiv.2310.13548},
  doi           = {10.48550/arXiv.2310.13548}
}

\appendix

\section{Dataset Details and Statistics}
\label{app:dataset_statistics}

To substantiate the empirical results presented in Section~\ref{sec:experiments} and facilitate reproducibility, we provide a granular description of the three evaluation datasets. These datasets were carefully curated to represent the distinct interaction topologies—adversarial, mixed, and open-ended—central to the BEACOF framework. A high-level statistical summary is provided in Table~\ref{tab:dataset_summary}.

\begin{table*}[t]
  \caption{Statistical summary of the datasets used in the evaluation scenarios. The ``Avg. Input Length'' denotes the average token count of the initial context provided to the agents.}
  \label{tab:dataset_summary}
  \begin{tabular}{lccccc}
    \toprule
    \textbf{Dataset} & \textbf{Scenario Type} & \textbf{Source} & \textbf{Samples} & \textbf{Key Modality} & \textbf{Avg. Input Length} \\
    \midrule
    Court Debate & Adversarial & AgentsCourt~\cite{DBLP:conf/emnlp/HeCWJ0XL0024} \& CJO & 100 & Unstructured Legal Text & $\sim$550 tokens \\
    MedQA & Mixed (Coop-Comp) & MedQA-USMLE~\cite{DBLP:journals/corr/abs-2009-13081} & 200 & Clinical Vignettes (5-Option) & $\sim$180 tokens \\
    Persona Chat & Open-Ended & PersonaChat~\cite{DBLP:conf/acl/KielaWZDUS18} & 100 & Dialogue History & $\sim$320 tokens \\
    \bottomrule
  \end{tabular}
\end{table*}

\subsection{Court Debate Dataset (Adversarial)}
The Court Debate dataset serves as a testbed for adversarial reasoning, comprising 100 criminal cases: 80 from the \emph{AgentsCourt} benchmark and 20 complex cases from \emph{China Judgements Online} (CJO). Unlike standard datasets, each entry captures a complete judicial reasoning cycle.

Formally, a case is modeled as a tuple $\mathcal{D} = (F, C_P, C_D, \mathcal{I}, Y)$. Here, $F$ details the factual narrative; $C_P$ and $C_D$ denote the initial claims of the Plaintiff and Defendant, respectively; $\mathcal{I}$ enumerates key adjudication issues (e.g., surrender claims); and $Y$ represents the ground truth verdict, including sentencing details. As detailed in Table~\ref{tab:court_distribution}, the dataset covers diverse domains to ensure robustness, spanning traffic offenses (e.g., Dangerous Driving, 28\%), property crimes (e.g., Theft, 15\%), and violent crimes (Intentional Injury, 10\%). This variety necessitates agent adaptability across distinct statutory contexts.

\begin{table}[H]
  \caption{Distribution of Case Causes in the Court Debate Dataset. The dataset covers a wide spectrum of criminal offenses to test adversarial robustness.}
  \label{tab:court_distribution}
  \resizebox{\columnwidth}{!}{%
  \begin{tabular}{lclc}
    \toprule
    \textbf{Case Cause} & \textbf{Count} & \textbf{Case Cause} & \textbf{Count} \\
    \midrule
    Dangerous Driving & 28 & Traffic Casualty & 12 \\
    Theft & 15 & Smuggling & 8 \\
    Fraud & 12 & Illegal Business/Gambling & 7 \\
    Intentional Injury & 10 & Others (e.g., Assault) & 8 \\
    \midrule
    \textbf{Total} & \multicolumn{3}{c}{\textbf{100}} \\
    \bottomrule
  \end{tabular}%
  }
\end{table}

\subsection{MedQA Dataset (Mixed)}
To evaluate professional consensus-building, we utilize the MedQA dataset, specifically sampling 200 questions from the MedQA-USMLE test set. This scenario embodies a mixed cooperative-competitive dynamic where agents must collaborate to diagnose but compete to eliminate incorrect distractors. 

The data structure for each instance consists of a complex clinical vignette, a set of five options (labeled A through E), and the ground-truth answer. The vignettes typically describe a patient's medical history, symptoms, and vital signs, requiring multi-hop reasoning to link physiological observations with pathological causes. The challenge lies in the high density of domain-specific terminology and the presence of plausible distractors, which compels agents to critically evaluate peer proposals rather than blindly accepting a consensus.

To mitigate potential bias from answer position, we verified the distribution of the ground-truth labels. As shown in Table~\ref{tab:medqa_distribution}, the correct answers are fairly uniformly distributed across the five options, ensuring that the agents' performance reflects genuine medical reasoning rather than statistical artifacts.

\begin{table}[H]
  \caption{Distribution of Correct Answer Keys in the MedQA Dataset. The balanced distribution prevents agents from exploiting positional bias.}
  \label{tab:medqa_distribution}
  \begin{tabular}{ccc}
    \toprule
    \textbf{Answer Key} & \textbf{Count} & \textbf{Percentage (\%)} \\
    \midrule
    A & 43 & 21.5 \\
    B & 40 & 20.0 \\
    C & 40 & 20.0 \\
    D & 49 & 24.5 \\
    E & 28 & 14.0 \\
    \midrule
    \textbf{Total} & \textbf{200} & \textbf{100.0} \\
    \bottomrule
  \end{tabular}
\end{table}

\subsection{Persona Chat Dataset (Open-Ended)}
For the open-ended social interaction scenario, we curated 100 interaction pairs from the \emph{PersonaChat} dataset to assess long-term consistency and lexical diversity. Unlike goal-oriented tasks, the objective here is maintaining a coherent identity.

Each sample is defined by two sets of \texttt{user\_personas}, which are lists of 3 to 5 sentences establishing the agent's background character (e.g., "I work as a stunt double," "I own a poodle"). Additionally, a \texttt{utterances} history is provided to prime the context. The evaluation metrics for this dataset are specifically designed to measure \emph{Persona Consistency}, using Natural Language Inference (NLI) to detect contradictions between generated responses and the assigned persona profile, and \emph{Response Diversity}, quantified via Distinct-N metrics.

A critical factor in social simulation is the depth of interaction. To demonstrate the complexity of the curated conversations, we analyze the distribution of dialogue lengths (measured in turns) in Table~\ref{tab:persona_length}. The majority of dialogues (67\%) exceed 20 turns, presenting a significant challenge for agents to maintain persona consistency over long contexts without succumbing to repetitive patterns.

\begin{table}[H]
  \caption{Distribution of Dialogue Lengths (Turns) in the Persona Chat Dataset. Long-context interactions challenge the agents' ability to maintain consistency.}
  \label{tab:persona_length}
  \begin{tabular}{ccc}
    \toprule
    \textbf{Dialogue Length (Turns)} & \textbf{Count} & \textbf{Percentage (\%)} \\
    \midrule
    11 -- 20 & 22 & 22.0 \\
    21 -- 30 & 49 & 49.0 \\
    31 -- 40 & 18 & 18.0 \\
    $> 40$   & 11 & 11.0 \\
    \midrule
    \textbf{Total} & \textbf{100} & \textbf{100.0} \\
    \bottomrule
  \end{tabular}
\end{table}

\section{Baselines}
\label{app:baselines}
To assess the effectiveness of belief-driven adaptation, we compare BEACOF against three representative frameworks that embody distinct positions on the interaction spectrum. First, representing Static Cooperation, we adopt CAMEL~\cite{DBLP:conf/nips/LiHIKG23}. This framework utilizes a role-playing inception prompting mechanism to facilitate task decomposition, serving as a strong baseline for scenarios requiring high coherence but often prone to sycophancy. Second, for Static Competition, we employ MAD~\cite{DBLP:conf/emnlp/Liang0JW00Y0T24}, which relies on "tit-for-tat" argumentation to surface errors. Specifically, in the Court Debate scenario, we utilize the AgentsCourt~\cite{DBLP:conf/emnlp/HeCWJ0XL0024} architecture, a domain-specialized implementation of MAD that strictly enforces the adversarial procedural logic of judicial systems. Additionally, we incorporate ReConcile~\cite{DBLP:conf/acl/ChenSB24} to represent Static Consensus-Building. Unlike pure cooperation, ReConcile employs a round-table voting mechanism with confidence refinement to aggregate diverse viewpoints. We employ ReConcile exclusively for the \textbf{MedQA} scenario. It is excluded from the Court Debate because its consensus-seeking objective fundamentally conflicts with the zero-sum nature of judicial proceedings, and from Persona Chat where open-ended diversity supersedes convergent solution-finding. To ensure a fair comparison and eliminate architectural bias, all methods utilize identical backbone LLMs, system prompts, and decoding parameters.

\section{Prompt Engineering Strategy}
\label{app:prompts}

To ensure the reproducibility of the BEACOF framework across diverse domains while maintaining a unified methodological presentation, we abstract the prompt engineering into a generalized schema. The framework operates on a dual-layer prompt architecture that separates the strategic coordination logic of the Meta-Agent from the execution logic of the Participant Agents. In implementation, domain-specific placeholders (e.g., \texttt{[DOMAIN\_CONTEXT]}) are instantiated with dataset-specific content (e.g., legal statutes for Court Debate, clinical vignettes for MedQA, or persona profiles for Persona Chat) at runtime.

\subsection{Meta-Agent Coordination Template}
The Meta-Agent functions as the mechanism designer, responsible for two critical tasks: (1) estimating the contextual payoff matrix to guide strategic equilibrium, and (2) evaluating participant outputs to update belief states. To support cross-scenario applicability, the prompt is structured to accept a variable set of evaluation dimensions $\mathcal{D} = \{d_1, \dots, d_n\}$ tailored to the specific task (e.g., \textit{Evidence Strength} for adversarial tasks or \textit{Empathy} for social tasks). The unified template is formalized as follows:

\begin{quote}
\ttfamily
\small
\textbf{System Instruction:} You are the Meta-Agent Coordinator overseeing a [SCENARIO\_TYPE] interaction. Your objective is to maintain the strategic equilibrium of the conversation.\\
\textbf{Contextual Input:} \\
- Global Context: [DOMAIN\_KNOWLEDGE\_BASE] \\
- Interaction History: [DIALOGUE\_HISTORY] \\
\textbf{Task 1 (Payoff Estimation):} Analyze the current state and estimate the potential utility for each participant if they adopt one of the following strategies: \textit{Cooperation}, \textit{Competition}, or \textit{Coopetition}. Assign a scalar value $u \in [0, 10]$ to each strategy-agent pair.\\
\textbf{Task 2 (Evaluation):} Assess the latest message based on the following domain-specific dimensions: [DIMENSION\_LIST]. For each dimension, provide a normalized score $s \in [0, 1]$ and a confidence score $\omega \in [0, 1]$.\\
\textbf{Output Requirement:} Return the results strictly in a structured JSON format containing keys for "payoff\_matrix" and "belief\_update\_vector".
\end{quote}

\subsection{Participant Agent Strategic Template}
Participant agents are designed to act as rational players within the incomplete information game. Unlike standard role-playing prompts that rely solely on persona descriptions, our template dynamically injects the game-theoretic signals computed by the PBE mechanism. This injection ensures that the agent's generative process is conditioned not only on its static role but also on the evolving belief states and strategic predictions. The generalized prompt template is defined below:

\begin{quote}
\ttfamily
\small
\textbf{Role Definition:} You are [AGENT\_ROLE], characterized by [PRIVATE\_PROFILE].\\
\textbf{Game State Injection:} \\
- \textbf{Current Beliefs:} Your subjective assessment of peers' capabilities is [BELIEF\_STATE]. \\
- \textbf{Strategic Signal:} The estimated payoffs for your potential actions are [PAYOFF\_MATRIX]. The predicted strategies of your opponents are [ACTION\_\\PREDICTION].\\
\textbf{Action Directive:} Based on the above information, you have resolved to adopt a [\textbf{SELECTED\_STRATEGY}] approach. \\
- If \textit{Cooperation}: Focus on information synthesis and consensus-building. \\
- If \textit{Competition}: Focus on critical argumentation and error exposure. \\
- If \textit{Coopetition}: Balance partial agreement with strategic rebuttal. \\
\textbf{Task:} Generate your response to [CURRENT\_QUERY] ensuring alignment with your selected strategy and private profile.
\end{quote}

By utilizing these templates, the framework standardizes the interaction flow across the Court Debate, MedQA, and Persona Chat scenarios, with the only variation being the semantic content of the bracketed placeholders.

\end{document}